\documentclass[aps,pra,superscriptaddress,twocolumn,longbibliography]{revtex4-2}
\usepackage{units}
\usepackage{amsmath}
\usepackage{amsthm}
\usepackage{amssymb}
\usepackage{graphicx}
\usepackage{color}
\usepackage{xcolor}
\usepackage{bbold}

\usepackage[colorlinks=true, urlcolor=blue, citecolor=blue,linkcolor=blue,citebordercolor={1 0 0},linkbordercolor={0 0 1}]{hyperref}

\newtheorem{theorem}{Theorem}
\newtheorem{lemma}[theorem]{Lemma}

\usepackage[linesnumbered, ruled,vlined]{algorithm2e}

\graphicspath{{./Images/},{./ImagesAppendix/}}

\renewcommand{\eqref}[1]{Eq.~(\ref{#1})} 

\theoremstyle{plain}

\theoremstyle{plain}

\ifx\proof\undefined
\newenvironment{proof}[1][\protect\proofname]{\par
	\normalfont\topsep6\p@\@plus6\p@\relax
	\trivlist
	\itemindent\parindent
	\item[\hskip\labelsep\scshape #1]\ignorespaces
}{%
	\endtrivlist\@endpefalse
}
\providecommand{\proofname}{Proof}
\fi
\theoremstyle{plain}

\theoremstyle{remark}

\newcommand{\bra}[1]{\langle #1|}
\newcommand{\ket}[1]{|#1 \rangle}
\makeatother
\newcommand{\braket}[2]{\langle #1 \vert #2 \rangle}
\newcommand{\abs}[1]{\left|#1\right|}
\newcommand{\idg}[1]{{\bfseries #1)}}

\newcommand\numberthis{\addtocounter{equation}{1}\tag{\theequation}}
\providecommand{\factname}{Fact}
\providecommand{\theoremname}{Theorem}
\providecommand{\claimname}{Claim}
\providecommand{\lemmaname}{Lemma}
\providecommand{\definitionname}{Definition}

\definecolor{KB}{rgb}{0.4,0.3,0.9}

\definecolor{THc}{rgb}{0.9,0.3,0.2}

\newcommand{\revA}[1]{{#1}}
\newcommand{\revB}[1]{#1}

\theoremstyle{definition}

\newcommand{\subfigimg}[3][,]{%
	\setbox1=\hbox{\includegraphics[#1]{#3}}
	\leavevmode\rlap{\usebox1}
	\rlap{\hspace*{2pt}\raisebox{\dimexpr\ht1-0.5\baselineskip}{{\bfseries \large\textsf{#2}}}}
	\phantom{\usebox1}
}

\newcommand{\sectionMain}[1]{
\let\oldaddcontentsline\addcontentsline
\renewcommand{\addcontentsline}[3]{}
\section{#1}
\let\addcontentsline\oldaddcontentsline
}

\allowdisplaybreaks

\begin{document}

\title{Natural parameterized quantum circuit}

\author{Tobias Haug}
\email{tobias.haug@u.nus.edu}
\affiliation{QOLS, Blackett Laboratory, Imperial College London SW7 2AZ, UK}
\author{M. S. Kim}
\affiliation{QOLS, Blackett Laboratory, Imperial College London SW7 2AZ, UK}
\begin{abstract}
Noisy intermediate scale quantum computers are useful for various tasks such as state preparation and variational quantum algorithms. However, the non-euclidean quantum geometry of parameterized quantum circuits is detrimental for these applications. 
Here, we introduce the natural parameterized quantum circuit (NPQC) that can be initialised with a euclidean quantum geometry. 
The initial training of variational quantum algorithms is substantially sped up as the gradient is equivalent to the quantum natural gradient.
Further, we show how to estimate the parameters of the NPQC by sampling the circuit, which could be used for benchmarking or calibrating NISQ hardware. For a general class of quantum circuits, the NPQC has the minimal quantum Cramér-Rao bound which highlights its potential for quantum metrology.  Finally, we show how to generate arbitrary superpositions of two states with the NPQCs for state preparation tasks.
Our results can be used to enhance currently available quantum processors.
\end{abstract}

\maketitle

\sectionMain{Introduction}
A growing number of applications for noisy intermediate scale quantum (NISQ) computers have been proposed~\cite{preskill2018quantum,bharti2021noisy} to make use of quantum computers available now and in the near future.
Variational quantum algorithms (VQAs) can help with tasks difficult for classical computers~\cite{peruzzo2014variational,kandala2017hardware,mcclean2016theory,cerezo2020variational} such as finding the ground state of Hamiltonians~\cite{peruzzo2014variational} or simulating quantum dynamics~\cite{otten2019noise,barison2021efficient}. A major obstacle for practical applications is the long training time of VQAs~\cite{lau2021quantum,bittel2021training,self2021variational}.
\revA{As further application, NISQ devices can be used to estimate parameters of the underlying quantum state~\cite{garciaperez2021learning,kaubruegger2021quantum,marciniak2021optimal}. A major challenge here is finding protocols and circuits that perform well in multi-parameter estimation tasks~\cite{szczykulska2016multi,meyer2020variational}.}

Parameterized quantum circuits (PQCs) are the basis of most NISQ algorithms. It is challenging to design PQCs that can efficiently run NISQ applications~\cite{haug2021capacity,nakaji2020expressibility,sim2019expressibility,du2020expressive}. The quantum geometry of PQCs as measured by the quantum Fisher information metric (QFIM) plays a key role in this regard~\cite{haug2021capacity,meyer2021fisher,katabarwa2021connecting}. VQAs can be trained more efficiently by using the QFIM for adaptive learning rates~\cite{haug2021optimal} and the quantum natural gradient (QNG)~\cite{haug2021optimal,stokes2020quantum,yamamoto2019natural,wierichs2020avoiding}.
For quantum sensing, the QFIM places a lower bound on the estimation error with the quantum Cramér-Rao bound~\cite{helstrom1976quantum,liu2019quantum,meyer2021fisher}. 
However, in general the QFIM of PQCs is not characterized, requires extensive resources to be a calculated~\cite{cerezo2021sub,gacon2021simultaneous,beckey2020variational,van2020measurement} and yields a non-euclidean geometry~\cite{haug2021capacity}, which is detrimental to tackle aforementioned tasks.

Here, we introduce the natural PQC (NPQC) which has a euclidean quantum geometry close to a particular reference parameter.
This expressive NPQC can be constructed in a hardware efficient manner even for many qubits and parameters, serving as a powerful basis for various NISQ applications.
The initial training iterations of VQAs with the NPQC are substantially improved as the gradient is equivalent to the QNG and we can use adaptive learning rates without needing to calculate the QFIM. We find that the first training step scales with increasing number of qubits which hints that our methods work even for larger systems.
\revA{Further, we demonstrate that the NPQC can prepare arbitrary superposition states of two states, a feature that can be useful for state preparation tasks.}
\revA{Finally, we show that by sampling the NPQC, one can estimate the absolute values of all the parameters of the circuit, which could be used for calibration purposes.
We also show that the NPQC has the minimal possible quantum Cramér-Rao bound for a general class of PQCs, which highlights its potential for multi-parameter metrology.}
These convenient properties make the NPQC a useful basis for various NISQ applications.

\begin{figure*}[htbp]
	\centering
	\subfigimg[width=0.9\textwidth]{}{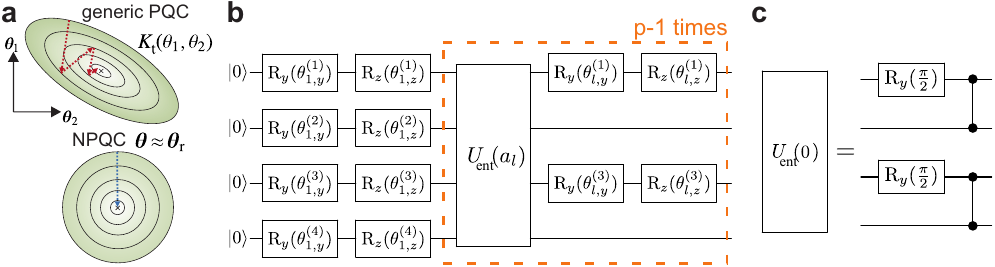}
	\caption{\idg{a} The fidelity landscape $K_\text{t}(\boldsymbol{\theta})=\abs{\braket{\psi(\boldsymbol{\theta}_\text{t})}{\psi(\boldsymbol{\theta}_\text{t}')}}^2$ as a function of parameter $\boldsymbol{\theta}$ of a parameterized quantum circuit (PQC). For generic PQCs, the landscape is non-euclidean and the parameters are distorted, which is characterized by the quantum Fisher information metric (QFIM) $\mathcal{F}(\boldsymbol{\theta})\ne cI$, where $c>0$ and $I$ is the identity matrix. Training with gradient ascent (dashed line) is challenging as the gradient does not point into the optimal direction.
	For a natural PQC (NPQC) the fidelity landscape is euclidean with $\mathcal{F}(\boldsymbol{\theta})=cI$. Thus, the standard gradient gives the optimal direction, is equivalent to the quantum natural gradient (QNG) and one can use adaptive learning rates to speed up training (\eqref{eq:update_add}).
	\idg{b}  Hardware efficient implementation of the NPQC composed of single qubit rotations and CPHASE gates with $N$ qubits.
	The NPQC has a euclidean quantum geometry at the reference parameter $\boldsymbol{\theta}\approx\boldsymbol{\theta}_\text{r}$ (\eqref{eq:ref_params}) with QFIM $\mathcal{F}(\boldsymbol{\theta}_\text{r})=I$. 
	\idg{c} Example of an entangling layer $U_\text{ent}(0)$ of the NPQC, consisting of $\frac{N}{2}$ CPHASE gates and single qubit rotations.
	}
	\label{fig:sketch}
\end{figure*}

\sectionMain{Model}
The QFIM $\mathcal{F}(\boldsymbol{\theta})$ for a PQC $\ket{\psi}=\ket{\psi(\boldsymbol{\theta})}$ and $M$-dimensional parameter vector $\boldsymbol{\theta}\in\mathbb{R}^M$ is an $M\times M$ dimensional positive semidefinite matrix~\cite{liu2019quantum,meyer2021fisher}
\begin{equation}\label{eq:quantumFisher}
\mathcal{F}_ {ij}(\boldsymbol{\theta})=4[\braket{\partial_i \psi}{\partial_j \psi}-\braket{\partial_i \psi}{\psi}\braket{\psi}{\partial_j \psi}]\,,
\end{equation} 
where $\partial_j\ket{\psi}$ is the gradient in respect to parameter $\theta_j$. 
The QFIM $\mathcal{F}(\boldsymbol{\theta})$ is a metric that relates fidelity of the quantum state with the distance in parameter space $\boldsymbol{\theta}$. When varying the parameter of the quantum state $\ket{\psi(\boldsymbol{\theta}+\text{d}\boldsymbol{\mu})}$ by a small $\text{d}\boldsymbol{\mu}$, the fidelity is given by 
\begin{equation}\label{eq:fidelityQFIM}
    \vert\braket{\psi(\boldsymbol{\theta})}{\psi(\boldsymbol{\theta}+\text{d}\boldsymbol{\mu})}\vert^2=1-\frac{1}{4}\text{d}\boldsymbol{\mu}^\text{T}\mathcal{F}(\boldsymbol{\theta})\text{d}\boldsymbol{\mu}\,.
\end{equation}
The variation of the distance in parameter space has a non-equal influence on the quantum state for generic PQCs with $\mathcal{F}(\boldsymbol{\theta})\ne cI$, where $I$ is the identity matrix and $c>0$. This non-euclidean nature of the PQC materializes in the QFIM, which acquires off-diagonal and unequal diagonal entries. 
A pictorial description of a non-euclidean fidelity landscape is shown in the upper graph of Fig.\ref{fig:sketch}a.
When the quantum geometry is euclidean with $\mathcal{F}(\boldsymbol{\theta})=cI$, then all the parameters $\boldsymbol{\theta}_i$, $\boldsymbol{\theta}_j$ are uncorrelated and they change the quantum state into orthogonal directions in the same proportional manner (see lower graph of Fig.\ref{fig:sketch}a). 
We define the NPQC as a PQC with a euclidean quantum geometry for a set of parameters. 

A hardware efficient construction of the NPQC is shown in Fig.\ref{fig:sketch}b. It consists of $N$ qubits ($N$ even) and $p$ layers of unitaries $U_l(\boldsymbol{\theta}_l)$ with quantum state $U(\boldsymbol{\theta})\ket{0}=\prod_{l=1}^p U_l(\boldsymbol{\theta}_l)\ket{0}^{\otimes N}$ parameterized by the $M$-dimensional parameter vector $\boldsymbol{\theta}\in\mathbb{R}^M$.
The first layer consists of $2N$ single qubit rotations around $y$ and $z$ axis applied on each qubit $n$ with $U_1=\prod_{n=1}^N R_z^{(n)}(\theta_{1,z}^{(n)})R_y^{(n)}(\theta_{1,y}^{(n)})$,
where $R_\alpha^{(n)}(\theta)=\exp(-i\frac{\theta}{2}\sigma^\alpha_n)$, $\alpha\in\{x,y,z\}$ and $\sigma^\alpha_n$ are the Pauli matrices applied on qubit $n$.
Each further layer ${l>1}$ is composed of a product of two qubit entangling gates and $N$ parameterized single qubit rotations given by
$U_l(a_l)=\prod_{k=1}^{N/2}[R_z^{(2k-1)}(\theta_{l,z}^{(2k-1)})R_y^{(2k-1)}(\theta_{l,y}^{(2k-1)})]U_\text{ent}(a_l)$, where $U_\text{ent}(a_l)=\prod_{k=1}^{N/2}\text{CPHASE}(2k-1,2k+2a_l)R_y^{(2k-1)}(\pi/2)$ and $\text{CPHASE}(n,m)$ is the controlled $\sigma^z$ gate applied on qubit index $n$, $m$, where indices larger than $N$ are taken modulo (see Fig.\ref{fig:sketch}c).
The shift factor $a_l\in\{0,1,\dots,N/2-1\}$ as a function of layer $l$ is defined via the following recursive rule.
Initialise the set $A=\{0,1,\dots,N/2-1\}$ and $s=1$. In each iteration, pick and remove one element $r$ from $A$. 
Then set $a_s=r$ and $a_{s+q}=a_{q}$ for $q=\{1,\dots,s-1\}$. As the last step, we set $s=2s$. We repeat this procedure until no elements are left in $A$ or the desired depth $p$ is reached. 
Our construction has up to $p_\text{max}=2^{N/2}$ layers with in total $M=N(p+1)$ parameters. 
The NPQC has a euclidean geometry where the QFIM is the identity for the reference parameter $\boldsymbol{\theta}_\text{r}$ given by
\begin{equation}\label{eq:ref_params}
\mathcal{F}(\boldsymbol{\theta}_\text{r})=I\hspace{0.3cm}\text{for}\hspace{0.3cm} \theta_{\text{r},l,y}^{(n)}=\pi/2\,,\hspace{0.3cm}\theta_{\text{r},l,z}^{(n)}=0 \,.
\end{equation}
\revA{The QFIM being identity means that variations of the parameters are independent of each other and lead to an orthogonal change in the space of quantum states.}
\revB{We numerically checked the QFIM $\mathcal{F}(\boldsymbol{\theta}_\text{r})=I$ for the NPQC for all $p$ and up to $N=14$ qubits, and given its regular structure we believe it will apply for any $N$. }
While the euclidean geometry is exactly valid only for $\boldsymbol{\theta}_\text{r}$, we find that it remains nearly euclidean in the vicinity of $\boldsymbol{\theta}\approx\boldsymbol{\theta}_\text{r}$. 
The QFIM $\mathcal{F}(\ket{\psi})=\mathcal{F}(V\ket{\psi})$ is invariant under application of arbitrary unitaries $V$~\cite{meyer2021fisher}.
Thus, the euclidean quantum geometry is preserved even if we apply additional unitaries on the NPQC.
We can prepare arbitrary reference states $\ket{\psi_\text{r}}=V_\text{ref}\ket{0}=\ket{\psi(\boldsymbol{\theta}_\text{r})}$ with the unitary $V_\text{ref}$
\begin{equation}\label{eq:initial_npqc}
\ket{\psi(\boldsymbol{\theta})}=V_\text{ref}U_\text{fix}^\dagger U(\boldsymbol{\theta})\ket{0}\,,
\end{equation}
where $U_\text{fix}=U(\boldsymbol{\theta}_\text{r})$ such that $U_\text{fix}^\dagger U(\boldsymbol{\theta}_\text{r})=I$. 
In general the NPQC is intractable for classical computers, however for the particular case $\boldsymbol{\theta}_\text{r}$ and $V_\text{ref}=I$ the NPQC is Clifford with an efficient simulation on classical computers~\cite{aaronson2004improved}. 

\revB{With our construction, the NPQC can yield up to $M\le N(2^{N/2}+1)$ parameters. We note that it is possible to extend the NPQC up to $2^{N+1}-2$ parameters such that all possible quantum states can be expressed~\cite{haug2021capacity}. However, we note that in this case a lower number of parameters per layer is achieved, and we were unable to find a general way to construct the circuit. }

\revB{\sectionMain{Expressibility and trainability}
First, we study the expressibility and trainability of the NPQC. To this end, we initialise the circuit with random parameters $\boldsymbol{\theta}_\text{rand}\in[0,2\pi]$ for different qubit number $N$ and depth $p$. 

Expressibility measures how well random instances of the circuit sample uniformly the Hilbertspace. We measure the expressibility with the frame potential~\cite{sim2019expressibility}
\begin{equation}
F_t=\int_{\boldsymbol{\theta}}\int_{\boldsymbol{\phi}} \vert\braket{\psi(\boldsymbol{\theta})}{\psi(\boldsymbol{\phi})}\vert^{2t} \text{d}\boldsymbol{\theta}\text{d}\boldsymbol{\phi}
\end{equation}
which measures the closeness to a $t$-design, i.e. how well random instances of the circuit approximate Haar random unitaries up to $t$th order. For Haar random unitaries and $t=2$, we have the minimal value $F_2^{\text{Haar}}=((2^N+1)2^{N-1})^{-1}$. We compute the expressibility $F_2$ by averaging over the square of the fidelity for randomly sampled instances of the circuit.

The trainability is measured with the variance of the gradients $\text{var}(\nabla E)$ in respect to a cost function $E$~\cite{mcclean2018barren}. For deep circuits, for many types of circuits the variance of the gradient decays exponentially with number of qubits, which is called barren plateaus. As the gradients are too small to be measured, barren plateaus are not trainable.
As cost function, we choose here $E=\bra{\psi}\sigma^z_1\sigma^z_2\ket{\psi}$. Note that the exact form of a local cost function has only negligible effect on the variance of the gradient~\cite{mcclean2018barren}.

As reference, we compare the performance of the NPQC with another hardware efficient circuit.
We choose the YZ-CNOT circuit, which is composed of layers of random parameterized $R_y$ and $R_z$ rotations with an entangling layer $U_\text{n.n-CNOT}$ of CNOT gates arranged in a nearest-neighbor chain configuration   $\ket{\psi_{\text{YZ-CNOT}}(\boldsymbol{\theta})}=\prod_{n=1}^p\prod_{k=1}^N R_z^{(k)}(\theta_{n,z}^{(k)})R_y^{(k)}(\theta_{n,y}^{(k)})U_\text{n.n-CNOT}$.

In Fig.\ref{fig:expressbility}a we study the variance of the gradient $\text{var}(\nabla E)$ against number of parameters $M$ of the circuit. We find that the variance decreases exponentially with $M$ and converges to a constant for sufficiently large $M$. We find that the variance is much larger for the NPQC compared to the YZ-CNOT circuit. 
In Fig.\ref{fig:expressbility}b we show the frame potential $F_2$ against $M$. We find that it decreases with $M$ and converges to a constant. While YZ-CNOT converges to the Haar random value, the NPQC has a larger $F_2$. 
In Fig.\ref{fig:expressbility}c, we plot the variance of the gradient against $N$ for deep circuits, such that variance and $F_2$ have converged. We find that the gradient decays exponentially for both NPQC and YZ-CNOT, with the NPQC showing a much slower descent. This implies that the NPQC can remain trainable even for relatively large $N$ compared to other hardware efficient circuits. 
In Fig.\ref{fig:expressbility}d, we plot the frame potential $F_2$. While the YZ-CNOT has a frame potential matching a Haar random unitary, the NPQC has larger $F_2$.

}
\begin{figure}[htbp]
	\centering	
	\subfigimg[width=0.24\textwidth]{a}{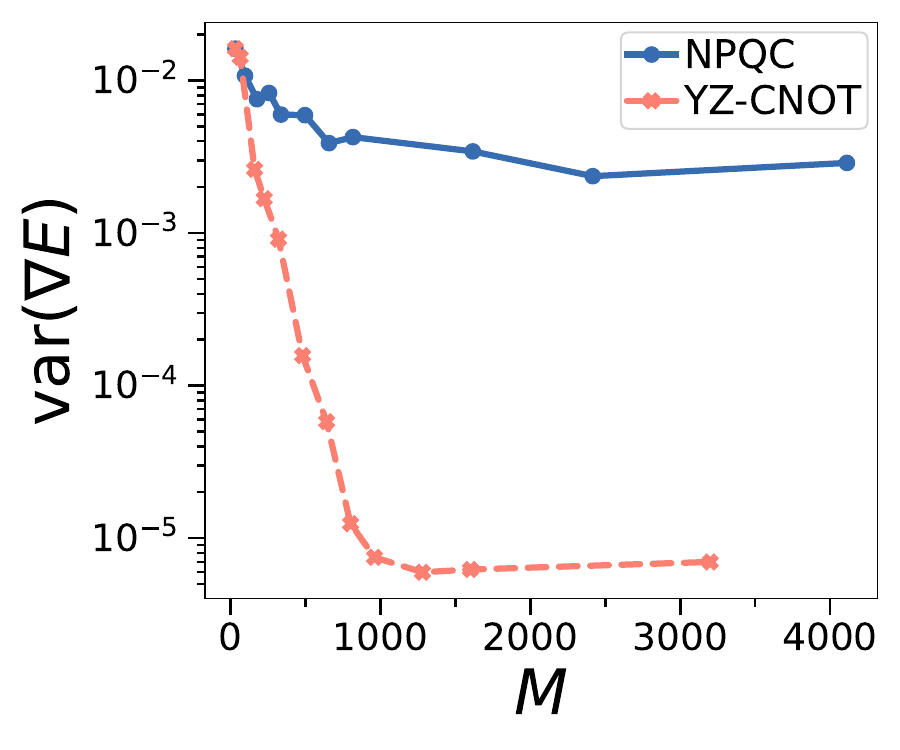}\hfill
	\subfigimg[width=0.24\textwidth]{b}{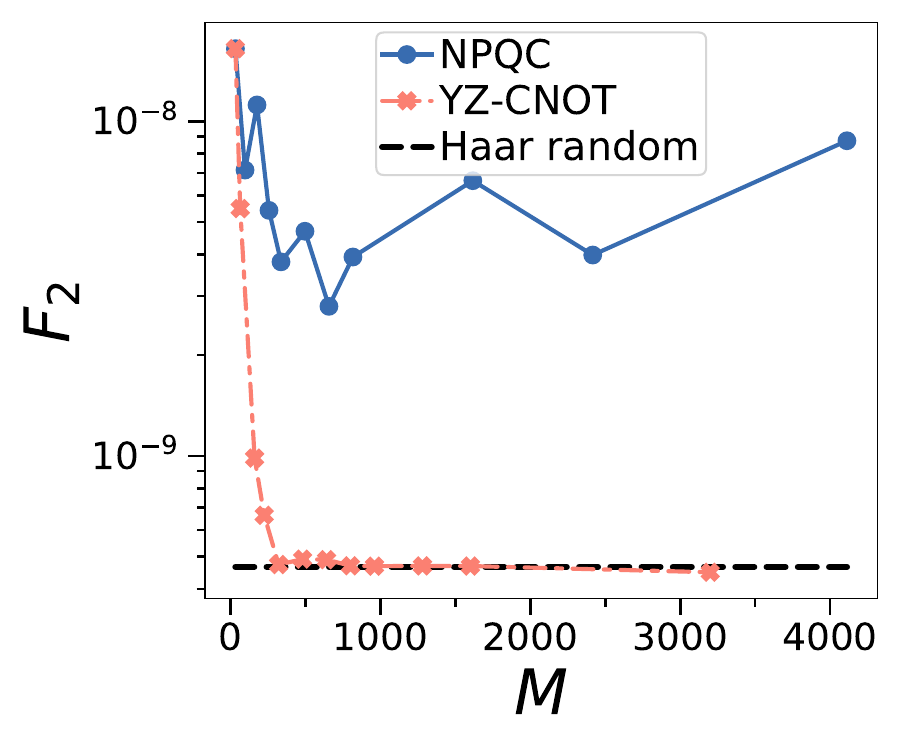}
	\subfigimg[width=0.24\textwidth]{c}{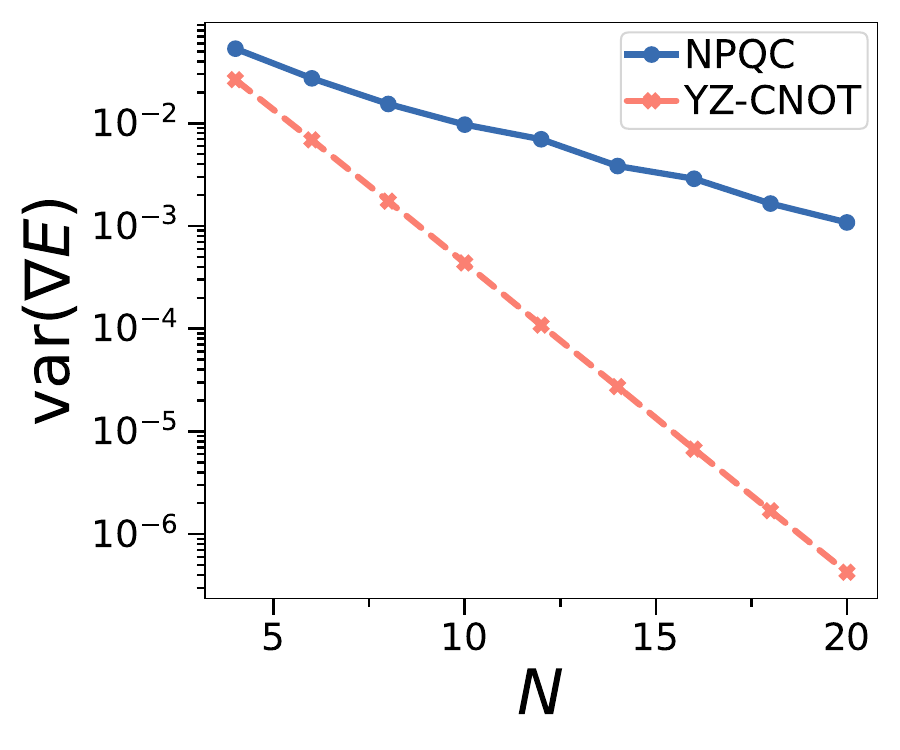}\hfill
	\subfigimg[width=0.24\textwidth]{d}{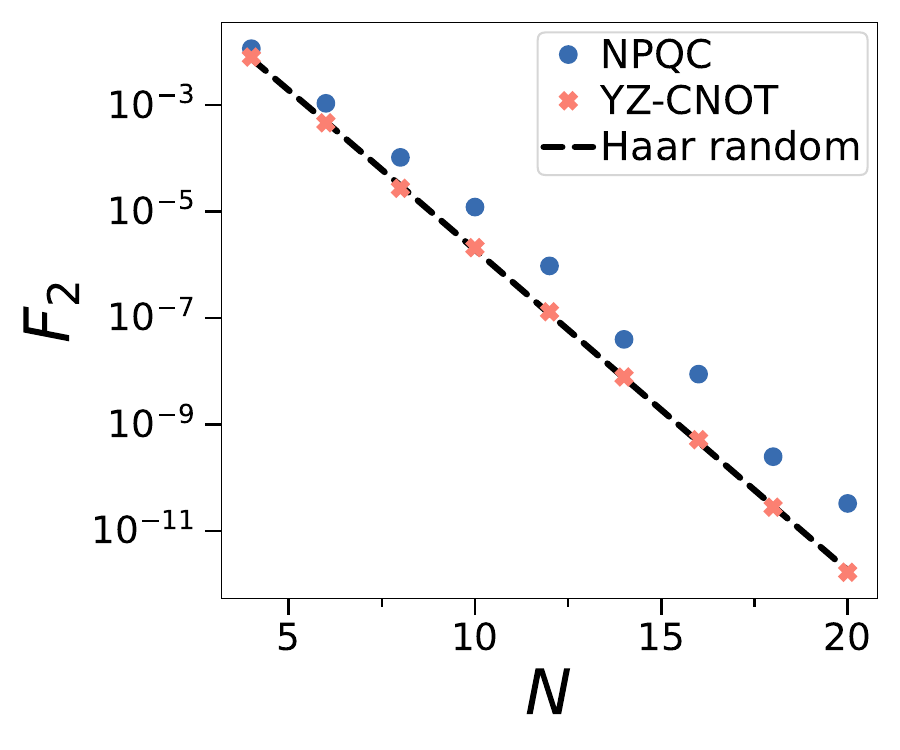}
	\caption{Expressibility and trainability of the NPQC and YZ-CNOT circuit as defined in main text. \idg{a} Variance of gradient $\text{var}(\nabla E)$ for $E=\bra{\psi}\sigma^z_1\sigma^z_2\ket{\psi}$ against number of parameters $M$ for qubit number $N=16$ \idg{b} Frame potential $F_2$ against $M$. 
	\idg{c} $\text{var}(\nabla E)$  against number of qubits $N$. We choose deep circuits with $p=2^{N/2}$ for the NPQC and $p=5N$ for the YZ-CNOT circuit. \idg{d} $F_2$ against $N$.
	We compute all values over 100 random instances of the circuits. }
	\label{fig:expressbility}
\end{figure}

Next, we investigate different applications of the NPQC for speeding up the training VQAs, preparing arbitrary superposition states, and multi-parameter metrology.

\sectionMain{Training VQAs}
VQAs solve tasks by optimizing the parameters $\boldsymbol{\theta}$ of the PQC in respect to a cost function. \revA{Various types of cost function have been studied such as the energy or fidelity. Due to the close connection of QFIM to fidelity, we concentrate here on the problem of learning a quantum state $\ket{\psi_\text{t}}$, an important subroutine in many VQAs~\cite{otten2019noise,benedetti2019generative,barison2021efficient,gibbs2021long}.}
The goal is to learn the target parameters $\boldsymbol{\theta}_\text{t}=\text{argmax}_{\boldsymbol{\theta}}K_\text{t}(\boldsymbol{\theta})$ by maximizing the fidelity
\begin{equation}
K_\text{t}(\boldsymbol{\theta})=\abs{\braket{\psi_\text{t}}{\psi(\boldsymbol{\theta})}}^2\,.
\end{equation}
We optimize the parameters iteratively with gradient ascent~\cite{peruzzo2014variational} via $\boldsymbol{\theta}'=\boldsymbol{\theta}+\alpha \nabla K_\text{t}(\boldsymbol{\theta})$, where $\alpha$ is the learning rate and $\nabla K_\text{t}(\boldsymbol{\theta})$ is the gradient, which points in the direction of steepest change of the cost function. 
As seen in Fig.\ref{fig:sketch}a, the gradient $\nabla K_\text{t}(\boldsymbol{\theta})$ is not the best choice for optimization as it implicitly assumes that the landscape is euclidean~\cite{amari2016information,stokes2020quantum,yamamoto2019natural}. 
To amend the non-euclidean nature, one can transform the gradient into the QNG ($\mathcal{F}^{-1}(\boldsymbol{\theta})\nabla K_\text{t}(\boldsymbol{\theta})$) by using the inverse of the QFIM $\mathcal{F}^{-1}(\boldsymbol{\theta})$~\cite{stokes2020quantum}. 
However, this transformation requires knowledge about the QFIM which can be difficult to acquire~\cite{amari2016information,stokes2020quantum,yamamoto2019natural}. 
As the NPQC has a euclidean geometry for the reference parameter ($\mathcal{F}^{-1}(\boldsymbol{\theta}_\text{r})=I$, \eqref{eq:ref_params}), the gradient and the QNG are equivalent 
\begin{equation}
    \nabla K_\text{t}(\boldsymbol{\theta}_\text{r})=\mathcal{F}^{-1}(\boldsymbol{\theta}_\text{r})\nabla K_\text{t}(\boldsymbol{\theta}_\text{r})\,,
\end{equation}
allowing us to perform the first training step with optimal geometry without needing to compute the QFIM (see lower graph of Fig.\ref{fig:sketch}a). 
To further improve training, we can replace the heuristic learning rate $\alpha$ with adaptive learning rates $\alpha_\text{t}(\boldsymbol{\theta})$ that change during the training. It has been shown that the fidelity $K$ of hardware efficient PQCs takes an approximate Gaussian form~\cite{haug2021optimal}
\begin{equation}\label{eq:kernel}
\mathcal{K}(\boldsymbol{\theta},\boldsymbol{\theta}')=\abs{\braket{\psi(\boldsymbol{\theta})}{\psi(\boldsymbol{\theta}')}}^2\approx\text{exp}[-\frac{1}{4}\Delta \boldsymbol{\theta}^{\text{T}}\mathcal{F}(\boldsymbol{\theta})\Delta \boldsymbol{\theta}]\,.
\end{equation}
Close to the reference parameter $\boldsymbol{\theta}\approx\boldsymbol{\theta}_\text{r}$, we have  $\mathcal{F}(\boldsymbol{\theta}\approx\boldsymbol{\theta_\text{r})}\approx I$. Together with \eqref{eq:kernel}, we find that the best choice of learning rate $\alpha_\text{t}(\boldsymbol{\theta})$ is given by (see Appendix~\ref{sec:gradient_sup} or~\cite{haug2021optimal})
\begin{align*}
\boldsymbol{\theta}_1=&\boldsymbol{\theta}+\alpha_1 \nabla K_\text{t},\hspace{0.4cm}\alpha_1=\frac{2\sqrt{-\log(K_\text{t}(\boldsymbol{\theta}))}}{\abs{\nabla K_\text{t}(\boldsymbol{\theta})}}\\
\alpha_\text{t}(\boldsymbol{\theta})=&\frac{2}{\alpha_1 \abs{\nabla K_\text{t}(\boldsymbol{\theta})}^2}\log\left(\frac{K_\text{t}(\boldsymbol{\theta}_1)}{K_\text{t}(\boldsymbol{\theta})}\right)+\frac{\alpha_1}{2}\,.\numberthis \label{eq:update_add}
\end{align*}
The adaptive learning rates combined with the inherent QNG can improve the training of VQAs. We initialise the NPQC with parameter $\boldsymbol{\theta}_\text{r}$ and choose any desired initial state via \eqref{eq:initial_npqc}. Then, we proceed to train the VQA for a few iterations with the adaptive learning rates.
After a few training iterations, the parameter of the NPQC $\boldsymbol{\theta}$ may not be close to $\boldsymbol{\theta}_\text{r}$ anymore and the QFIM can acquire substantial off-diagonal entries. At this point, our assumption $\mathcal{F}=I$ breaks down and we switch to a heuristic learning rate. Nonetheless, improving the initial training iterations can already give us a speed up in training VQAs.

\begin{figure}[htbp]
	\centering	
	\subfigimg[width=0.3\textwidth]{a}{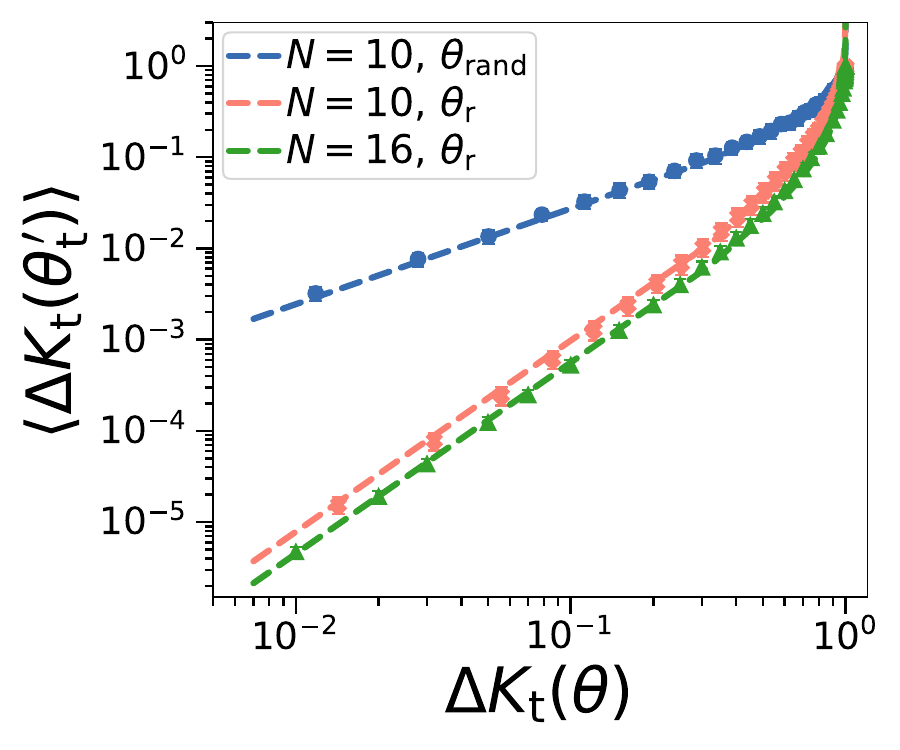}\hfill
	\subfigimg[width=0.3\textwidth]{b}{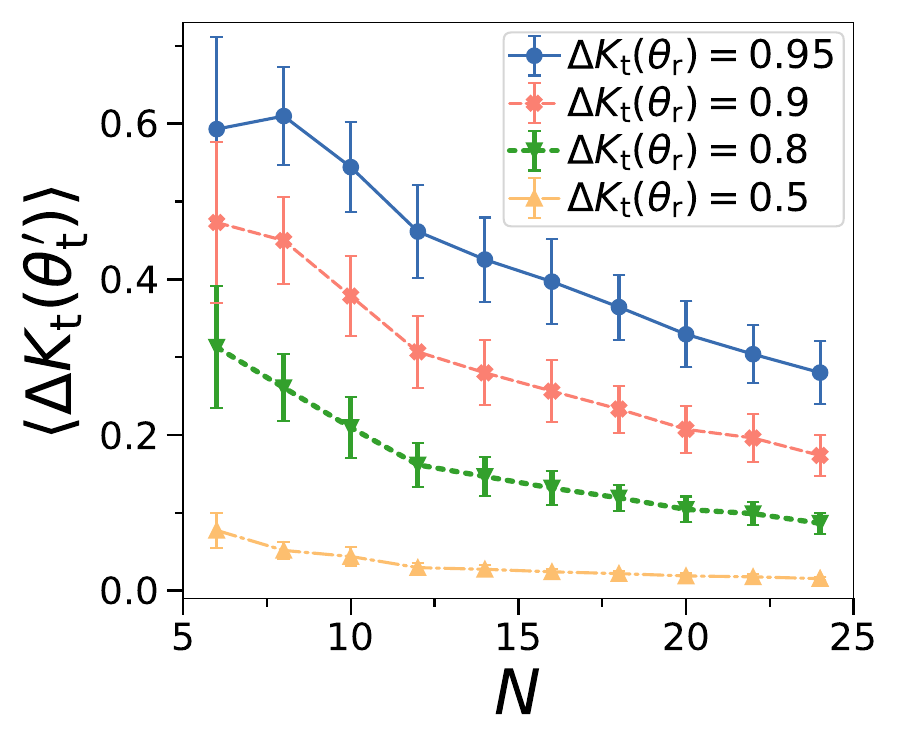}
	\caption{\idg{a} Average infidelity $\langle\Delta K_\text{t}(\boldsymbol{\theta}_\text{t}')\rangle$ after one iteration of gradient ascent with adaptive learning rates. The initial parameters of the NPQC are random $\boldsymbol{\theta}_\text{rand}\in[0,2\pi]$ ($\mathcal{F}(\boldsymbol{\theta}_\text{rand})\ne I$, blue curve) or the reference parameter $\boldsymbol{\theta}_\text{r}$ ($\mathcal{F}(\boldsymbol{\theta}_\text{r})=I$, orange and green curves). We show $\langle\Delta K_\text{t}(\boldsymbol{\theta}_\text{t}')\rangle$ against initial infidelity $\Delta K_\text{t}(\boldsymbol{\theta})$. Dashed lines are fits with $\Delta K_\text{t}(\boldsymbol{\theta}_\text{t}') =-c\log^\nu[1- \Delta K_\text{t}(\boldsymbol{\theta}) ]$ with $\nu=1$ for $\boldsymbol{\theta}_\text{rand}$ and $\nu=2$ for $\boldsymbol{\theta}_\text{r}$.
	The scaling factors of the fits are $c=\{6\cdot 10^{-2},4.7\cdot 10^{-3},2.7\cdot 10^{-3}\}$ and number of qubits $N=10$.
    \idg{b} Average infidelity after a single iteration of gradient ascent $\langle \Delta K_\text{t}(\boldsymbol{\theta}_\text{t}')\rangle$ plotted against number of qubits $N$ for varying infidelity before the step $\Delta K_\text{t}(\boldsymbol{\theta}_\text{r})$. Number of layers is $p=10$ and data is averaged over 50 random instances \revB{for both figures, where the error bars show the standard deviation.}
	}
	\label{fig:training_single}
\end{figure}
Here, we numerically study the performance of NPQCs for learning a target quantum state $\ket{\psi_\text{t}}=\ket{\psi(\boldsymbol{\theta}_{\text{t}})}$~\cite{yao,johansson2012qutip}. \revB{We measure the quality of the found target parameter $\boldsymbol{\theta}_\text{t}'$ with the infidelity
$\Delta K_\text{t}(\boldsymbol{\theta}_\text{t}')=1-K_\text{t}(\boldsymbol{\theta}_\text{t}')$ between trained state $\ket{\psi(\boldsymbol{\theta}_{\text{t}}')}$ and the actual target state. 
We generate random target parameters $\boldsymbol{\theta}_\text{t}=\boldsymbol{\theta}_\text{ini}+\Delta \boldsymbol{\theta}$ by shifting the initial parameters $\boldsymbol{\theta}_\text{ini}$ with a randomly chosen $\Delta \boldsymbol{\theta}$, where we achieve a desired initial infidelity $\Delta K_\text{t}(\boldsymbol{\theta}_\text{ini})$ by choosing the norm of $\Delta \boldsymbol{\theta}$ according to \eqref{eq:kernel}. }
First, in Fig.\ref{fig:training_single} we study the performance of a single iteration of gradient descent.
In Fig.\ref{fig:training_single}a, we plot the infidelity after one iteration of gradient ascent $\langle \Delta K_\text{t}(\boldsymbol{\theta}_\text{t}') \rangle$ using adaptive learning rates for varying initial infidelity $\Delta K_\text{t}(\boldsymbol{\theta})$. We compare training starting with random initial parameters of the NPQC $\boldsymbol{\theta}=\boldsymbol{\theta_\text{rand}}\in[0,2\pi]$ ($\mathcal{F}(\boldsymbol{\theta}_\text{rand})\ne I$) against the reference parameter $\boldsymbol{\theta}_\text{r}$ with euclidean quantum geometry ($\mathcal{F}(\boldsymbol{\theta}_\text{r})=I$). Training with the euclidean starting point $\boldsymbol{\theta}_\text{r}$ outperforms the randomly chosen parameters.
In Fig.\ref{fig:training_single}b, we observe that the infidelity after one iteration of gradient ascent $\langle\Delta K_\text{t}(\boldsymbol{\theta}_\text{t}')\rangle$ decreases with increasing qubit number $N$, demonstrating improved performance when scaling up the number of qubits.

\begin{figure}[htbp]
	\centering	
	\subfigimg[width=0.3\textwidth]{a}{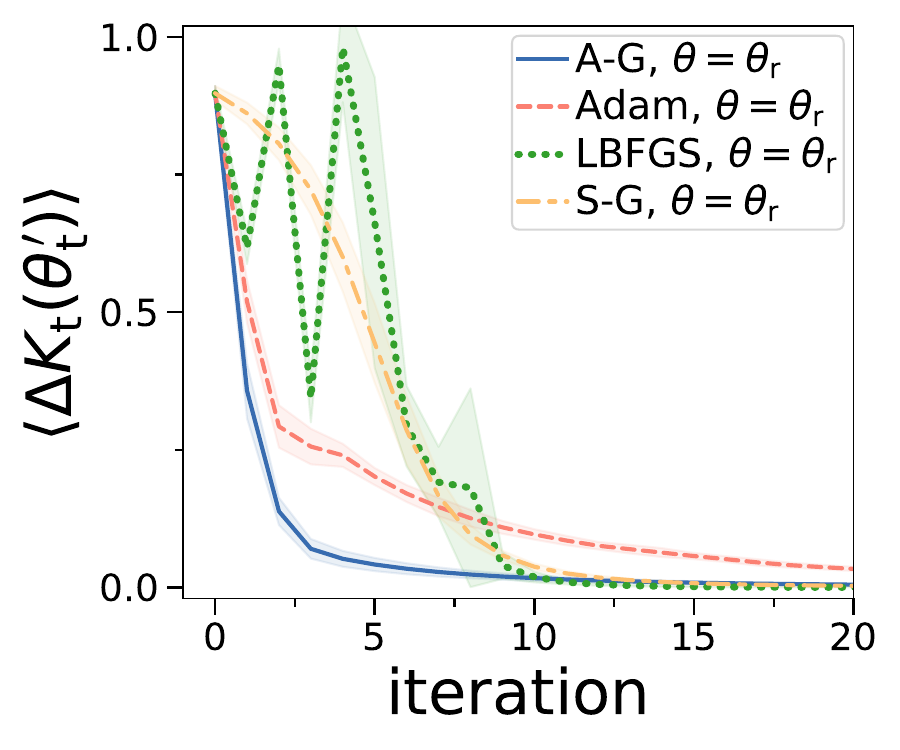}\hfill
	\subfigimg[width=0.3\textwidth]{b}{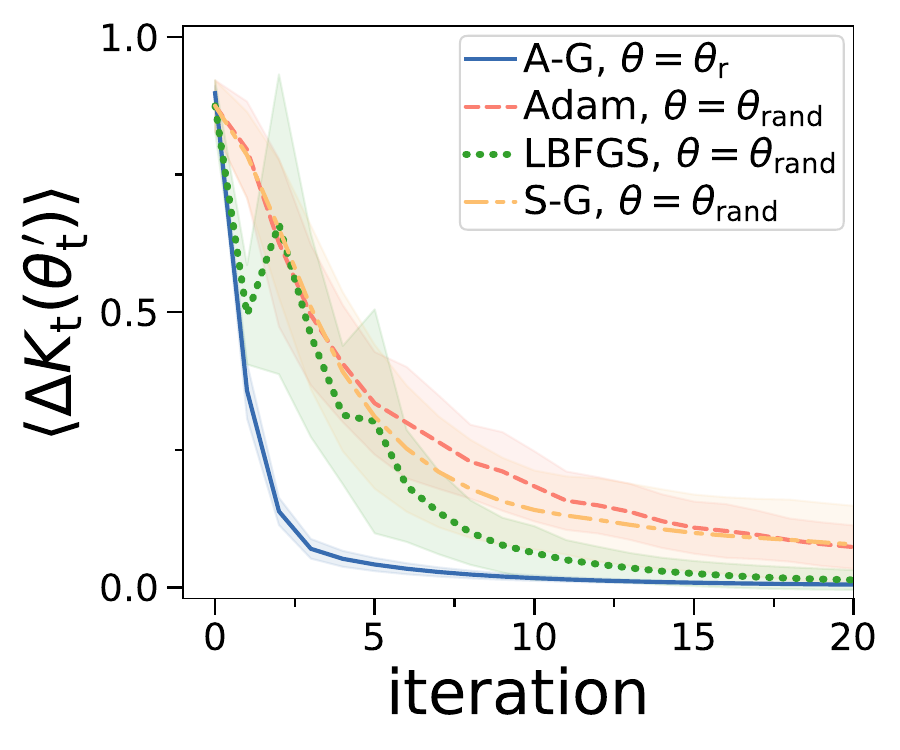}
	\caption{ 
	\idg{a} Training NPQC with initial parameter $\boldsymbol{\theta}_\text{r}$ and initial infidelity $\Delta K_\text{t}(\boldsymbol{\theta}_\text{r})=0.9$. Shaded area is standard deviation of infidelity over 50 random instances of the target state. We compare adaptive gradient ascent (A-G, adaptive learning rates for first 3 iterations, then fixed learning rate $\alpha=0.5$), Adam, LBFGS and standard gradient ascent (S-G, learning rate $\alpha=1$). 
	\idg{b} Training starting with reference parameter $\boldsymbol{\theta}_\text{r}$ and random parameter $\boldsymbol{\theta}_\text{rand}\in[0,2\pi]$ with $\Delta K_\text{t}(\boldsymbol{\theta}_\text{rand})=0.9$. 
	}
	\label{fig:training}
\end{figure}

Next, we study training over multiple iterations in Fig.\ref{fig:training}.
In Fig.\ref{fig:training}a, we show training with the NPQC using various optimization methods using $\boldsymbol{\theta}_\text{r}$ as initial parameter. We compare adaptive gradient ascent (A-G) with standard methods such as Adam~\cite{kingma2014adam}, LBFGS~\cite{fletcher2013practical} and standard gradient ascent with fixed learning rate (S-G). For adaptive gradient ascent, we use~\eqref{eq:update_add} for the first three training iterations, then switch to a heuristic learning rate as the QFIM becomes non-euclidean. 
In Fig.\ref{fig:training}b, we compare training with $\boldsymbol{\theta}_\text{r}$ as initial parameter against training with random initial parameters $\boldsymbol{\theta}_\text{rand}\in[0,2\pi]$. 
We find that adaptive gradient ascent performs superior compared to the others methods as the first training iterations can leverage the euclidean quantum geometry to provide a substantial speed up. We show further data on training VQAs in Appendix~\ref{sec:training_sub}.

\sectionMain{Generating superposition states}\label{sec:superposition_sup}
\revA{Next, we show that the special structure of the NPQC allows us prepare arbitrary superposition states of two states. The NPQC can generate superposition states of the reference state $\ket{\psi(\boldsymbol{\theta}_\text{r})}$ and a given target state $\ket{\psi(\boldsymbol{\theta}_\text{t})}$ with target parameters $\boldsymbol{\theta}_\text{t}$. We want to find the parameters $\boldsymbol{\theta}_\text{s}$ for the superposition state
\begin{equation}
\ket{\psi_\text{s}}=\ket{\psi(\boldsymbol{\theta}_\text{s})}=\gamma_\text{r}\ket{\psi(\boldsymbol{\theta}_\text{r})}+\gamma_\text{t}\ket{\psi(\boldsymbol{\theta}_\text{t})}+\gamma_\perp\ket{\psi_{\perp}}\,,
\end{equation}
where $\ket{\psi_{\perp}}$ is orthogonal to both reference and target states.
Now, the NPQC can generate superposition states $\ket{\psi_\text{s}}$ with tailored fidelities with the reference state $K_\text{r,s}=\abs{\braket{\psi(\boldsymbol{\theta}_\text{r})}{\psi_\text{s}}}^2$ and the target state $K_\text{t,s}=\abs{\braket{\psi(\boldsymbol{\theta}_\text{t})}{\psi_\text{s}}}^2$ to our choosing.
We define $\Delta \boldsymbol{\theta}_\text{r,s}=\boldsymbol{\theta}_\text{s}-\boldsymbol{\theta}_\text{r}$ as the difference between the parameters of the superposition state and target state, as well as $\Delta \boldsymbol{\theta}_\text{r,t}=\boldsymbol{\theta}_\text{t}-\boldsymbol{\theta}_\text{r}$ as the difference between the parameters of the target and reference state. 
Using \eqref{eq:kernel}, we calculate the relation between fidelity and parameter distance for the reference and superposition state
\begin{equation}
\abs{\Delta \boldsymbol{\theta}_\text{r,s}}^2=-4\log(K_\text{r,s})\,.
\end{equation}
For the target state and superposition state, we have
\begin{align*}
&K_\text{t,s}=e^{-\frac{1}{4}\abs{\boldsymbol{\theta}_\text{s}-\boldsymbol{\theta}_\text{t}}^2}=e^{-\frac{1}{4}\abs{\Delta\boldsymbol{\theta}_\text{r,s}-\boldsymbol{\theta}_\text{r,t}}^2}\\
&=e^{-\frac{1}{4}(\abs{\Delta\boldsymbol{\theta}_\text{r,s}}^2+\abs{\boldsymbol{\theta}_\text{r,t}}^2-2\abs{\Delta\boldsymbol{\theta}_\text{r,s}}\abs{\boldsymbol{\theta}_\text{r,t}}\cos(\measuredangle(\Delta \boldsymbol{\theta}_{r,s},\Delta \boldsymbol{\theta}_{r,t})))}\\
&=K_\text{r,s}e^{-\frac{1}{4}(\abs{\boldsymbol{\theta}_\text{r,t}}^2-4\sqrt{-\log(K_\text{r,s})}\abs{\boldsymbol{\theta}_\text{r,t}}\cos(\measuredangle(\Delta \boldsymbol{\theta}_{r,s},\Delta \boldsymbol{\theta}_{r,t}))}\,,
\end{align*}
where $\measuredangle(\Delta \boldsymbol{\theta}_\text{r,s},\Delta \boldsymbol{\theta}_\text{r,t})$ is the angle between the two parameter vectors.
By rearranging the equation and taking the logarithm, we finally get
\begin{equation}\label{eq:cos_sup}
\cos[\measuredangle(\Delta \boldsymbol{\theta}_\text{r,s},\Delta \boldsymbol{\theta}_\text{r,t})]=\frac{4\log\left(\frac{K_\text{t,s}}{K_\text{r,s}}\right)+\abs{\Delta \boldsymbol{\theta}_\text{r,t}}^2}{4\abs{\Delta \boldsymbol{\theta}_\text{r,t}}\sqrt{-\log(K_\text{r,s})}}\,.
\end{equation}
A solution exists when the absolute value of the right hand side of~\eqref{eq:cos_sup} is less or equal 1.
The boundary of the solution space is given by
\begin{equation}
K_\text{t,s}=K_\text{r,s}\exp(\pm\abs{\Delta \boldsymbol{\theta}_\text{r,t}}\sqrt{-\log(K_\text{r,s})}-\frac{1}{4}\abs{\Delta \boldsymbol{\theta}_\text{r,t}}^2)
\end{equation}
 We define the error between desired and actual superposition state
\begin{equation}\label{eq:error_superposition}
\Delta C=\abs{K_\text{r,s}-K_\text{r,s}'}+\abs{K_\text{t,s}-K_\text{t,s}'}\,,
\end{equation}
where $K_\text{r,s}'$ and $K_\text{t,s}'$ are the actual fidelities measured with reference and target state respectively, and $\Delta C=0$ corresponds to perfect creation of the desired superposition state. 

Now, we investigate generating superposition states with the NPQC. The superposition state $\ket{\psi_\text{s}}$ is a linear combination of reference state $\ket{\psi(\boldsymbol{\theta}_\text{r})}$ and random target state $\ket{\psi(\boldsymbol{\theta}_\text{t})}$ with infidelity between target and reference state $\Delta K_\text{t}(\boldsymbol{\theta}_\text{r})$. We randomly choose a desired fidelity $K_\text{t,s}$ between target and superposition state, and fidelity $K_\text{r,s}$ between reference and superposition state. Then, we calculate the parameters $\boldsymbol{\theta}_\text{s}$ of the superposition state using~\eqref{eq:cos_sup} and generate the state using the NPQC.
In Fig.\ref{fig:superposition_sup}a, we plot the error of the superposition states $\Delta C$ for different $K_\text{t,s}$ and $K_\text{r,s}$. The dashed line shows the boundary of possible superposition states.
In Fig.\ref{fig:superposition_sup}b, we show $\cos[\measuredangle(\Delta \boldsymbol{\theta}_\text{r,s},\Delta \boldsymbol{\theta}_\text{r,t})]$ as a function of the fidelities $K_\text{r,s}$ and $K_\text{t,s}$.
In Fig.\ref{fig:superposition_sup}c, we find that the error $\Delta C$ decreases with number of parameters $M$ of the NPQC and increases with $\Delta K_\text{t}(\boldsymbol{\theta}_\text{r})$.}

\begin{figure*}[htbp]
	\centering
	\subfigimg[width=0.3\textwidth]{a}{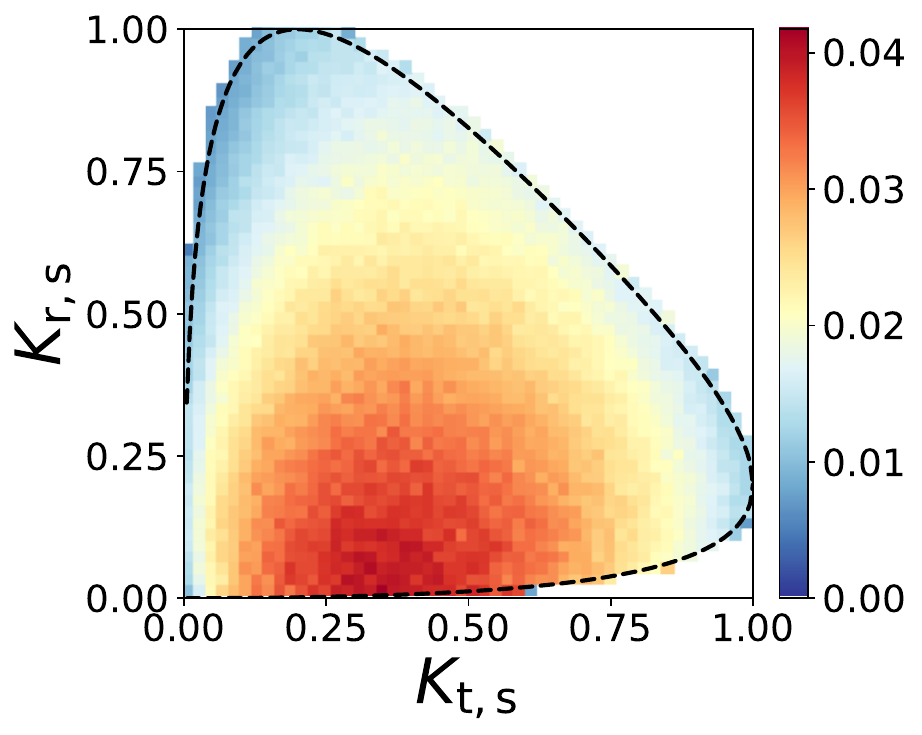}
	\subfigimg[width=0.31\textwidth]{b}{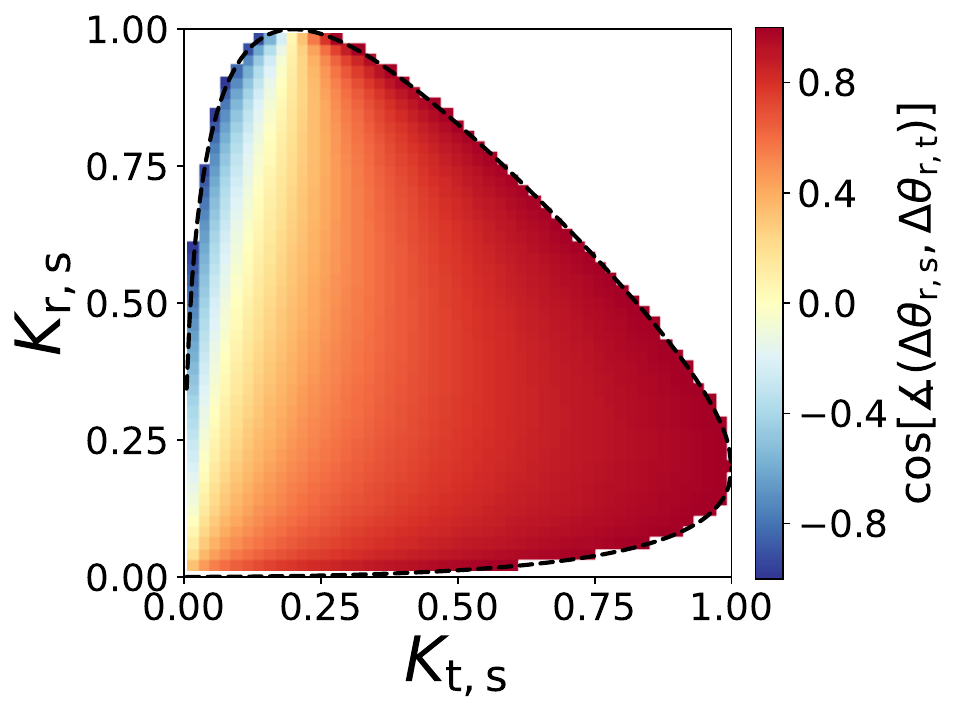}
	\subfigimg[width=0.3\textwidth]{c}{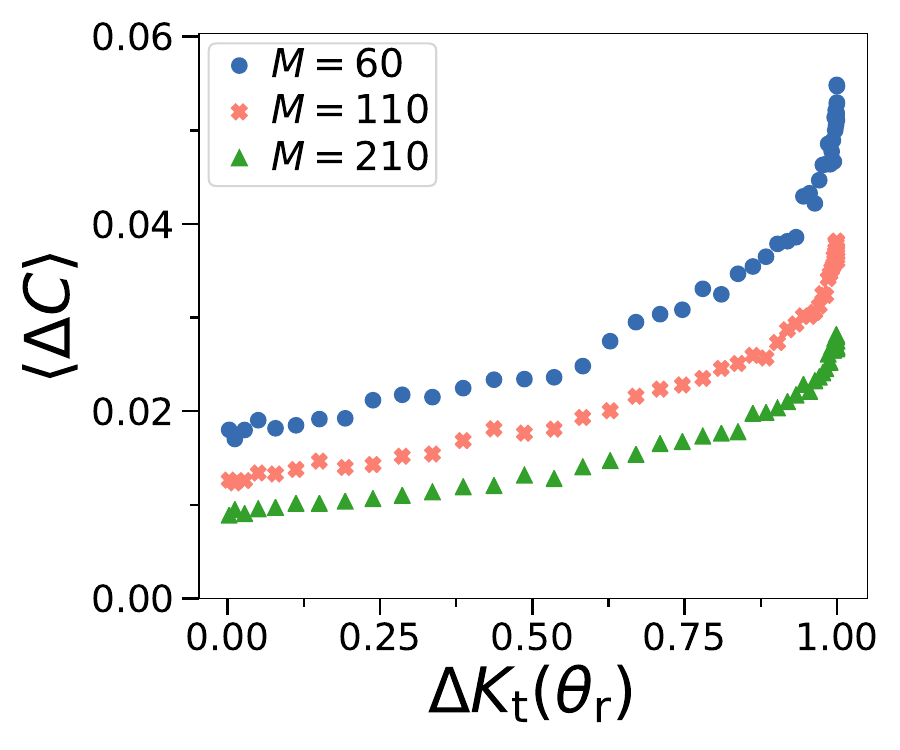}
	\caption{\idg{a} Average error $\langle\Delta C\rangle$ (\eqref{eq:error_superposition}) of generating superposition states  for $N=10$ qubits. We plot $\langle\Delta C\rangle$ as a function of the desired fidelity between target and superposition state $K_\text{t,s}$, as well as reference and superposition state $K_\text{r,s}$. Dashed line is the boundary of possible superposition states. The infidelity between reference and target state is $\Delta K_\text{t}(\boldsymbol{\theta}_\text{r})=0.8$ and the number of parameters of the NPQC is $M=110$.
	\idg{b} We show the relative angle between the difference vector of reference and superposition parameters, and reference and target parameters $\cos[\measuredangle(\Delta \boldsymbol{\theta}_\text{r,s},\Delta \boldsymbol{\theta}_\text{r,t})]$. 
	\idg{c} $\langle\Delta C\rangle$ as a function of infidelity $\Delta K_\text{t}(\boldsymbol{\theta}_\text{r})$ for varying $M$. $\langle\Delta C\rangle$ is averaged over 1000 random instances of fidelities $K_\text{t,s}$ and $K_\text{r,s}$.
	}
	\label{fig:superposition_sup}
\end{figure*}

\sectionMain{Parameter estimation}
\revA{As next application, we want to estimate the parameters of the NPQC by sampling from the quantum computer.
In particular, we want to estimate the entries of the $M$-dimensional vector $\Delta\boldsymbol{\theta}$ by performing measurements on the quantum state $\ket{\psi(\boldsymbol{\theta}_\text{r}+\Delta\boldsymbol{\theta})}$. Our task differs from standard quantum metrology, as we put two restrictions on the protocol. First, we only estimate the absolute values of each entries, i.e. $\vert\Delta\boldsymbol{\theta}_i\vert$, $i\in\{1,\dots,M\}$. Further, we assume that the magnitude of $\Delta\boldsymbol{\theta}$ is small. }

Using the NPQC, we can estimate these $M$ parameters by sampling in the computational basis only. The setup is a modified version of the NPQC $U_y(\boldsymbol{\theta})$, where all the parameterized single-qubit $z$-rotations are removed and we do not vary the $N/2$ $y$-rotations on the qubits with even index (see Appendix~\ref{sec:NPQC_constr_sup}). Then, we use $\ket{\psi_y(\boldsymbol{\theta})}=U_y^\dagger(\boldsymbol{\theta}_\text{r})U_y(\boldsymbol{\theta})\ket{0}$, where we apply the adjoint $U_y^\dagger(\boldsymbol{\theta}_\text{r})$.
A small variation $\abs{\Delta\boldsymbol{\theta}}\ll1$ yields
\begin{equation}\label{eq:sense_state}
\ket{\psi_y(\boldsymbol{\theta}_\text{r}+\Delta\boldsymbol{\theta})}\approx\sqrt{1-\frac{1}{4}\vert\Delta\boldsymbol{\theta}\vert^2}\ket{0}+\frac{1}{2}\sum_{i=1}^M (-1)^{\alpha_i}\Delta\boldsymbol{\theta}_i\ket{v_i}\,,
\end{equation}
where $\ket{v_i}$ is the computational basis state with the unique number $v_i$ for the $i$-th parameter of the NPQC and $\alpha_i\in \{0,1\}$. 
The approximate form of \eqref{eq:sense_state} is motivated in Appendix~\ref{sec:sensing_approx}.
The number $v_i$ can be efficiently determined on a classical computer from the gradients in respect to parameter $\Delta\boldsymbol{\theta}_i$ for the Clifford state $\ket{\psi_y(\boldsymbol{\theta}_\text{r})}$.
For small variations, the absolute value of the $i$-th parameter entry $\vert\Delta\boldsymbol{\theta}_i\vert$ can be estimated by sampling from $\ket{\psi_y(\boldsymbol{\theta}_\text{r}+\Delta\boldsymbol{\theta})}$ in the computational basis with $\vert\Delta\boldsymbol{\theta}_i\vert=2\sqrt{P_i}$, where $P_i=\abs{\braket{\psi_y(\boldsymbol{\theta}_\text{r}+\Delta\boldsymbol{\theta})}{v_i}}^2$ is the probability of measuring the computational basis state $\ket{v_i}$. As these measurements commute, one can determine $M=\frac{pN}{2}$ parameters at the same time for a NPQC with $p$ layers.

Now, we demonstrate our estimation protocol. In Fig.\ref{fig:sensing} we show the relative root mean square error (RMSE) to estimate the $M$-dimensional parameter vector $\Delta\boldsymbol{\theta}$ of the NPQC. We show in Fig.\ref{fig:sensing}a that the error decreases with increasing number of measurement samples $n$, reaching eventually a constant error. The error decreases when the parameter $\abs{\Delta\boldsymbol{\theta}}$ to be estimated becomes smaller as we derived our protocol in the limit of small $\abs{\Delta\boldsymbol{\theta}}$.
In Fig.\ref{fig:sensing}b, we show that for infinite number of measurements $n$ the error decreases to nearly zero with decreasing norm $\abs{\Delta\boldsymbol{\theta}}$ of the parameter vector to be estimated. For finite $n$, we observe that for small $\abs{\Delta\boldsymbol{\theta}}$ the error increases as the number of measurements is too low to reliably estimate the probability distribution of the computational basis states (see \eqref{eq:sense_state}). We observe that our protocol has a sweet spot where the relative error is minimal.

\begin{figure}[htbp]
	\centering
	\subfigimg[width=0.3\textwidth]{a}{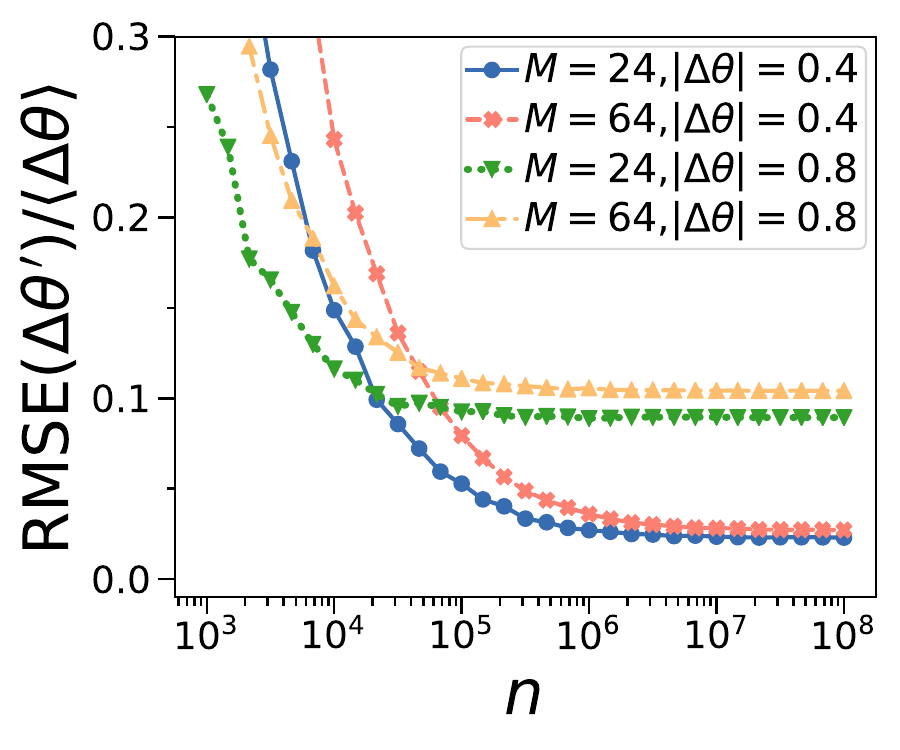}\hfill
	\subfigimg[width=0.3\textwidth]{b}{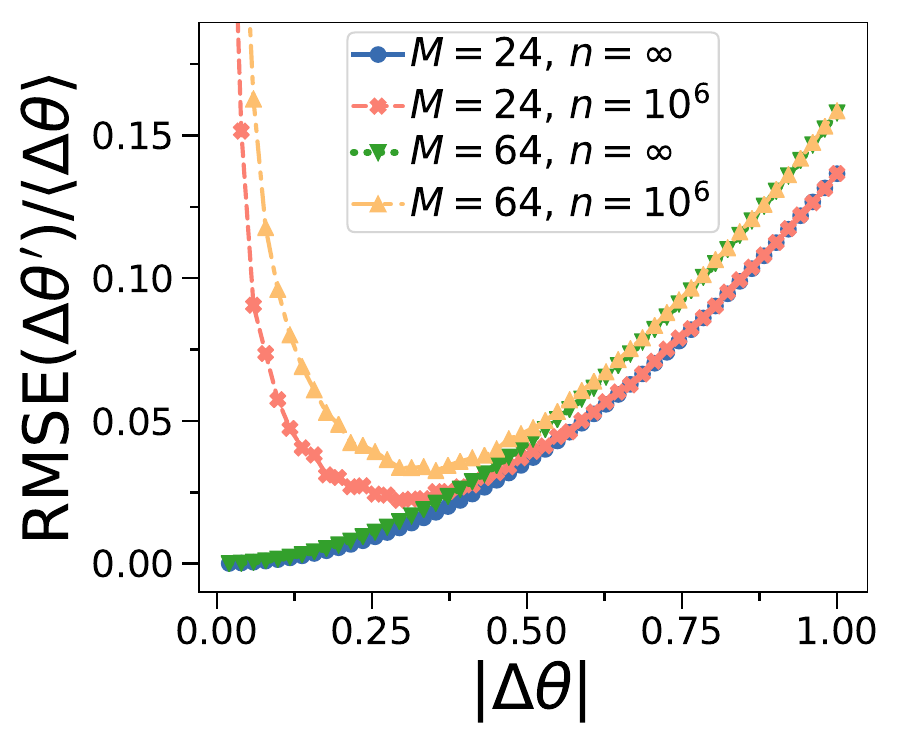}
	\caption{\idg{a} Estimating parameters $\Delta\theta$ by sampling from NPQC $\ket{\psi_y(\boldsymbol{\theta}_\text{r}+\Delta\boldsymbol{\theta})}$. We plot the root mean square error $\text{RMSE}(\Delta\boldsymbol{\theta}')=\sqrt{\langle(\Delta\boldsymbol{\theta}'-\Delta\boldsymbol{\theta})^2\rangle}$ of the estimated parameter $\Delta\boldsymbol{\theta}'$ normalized by the average parameter $\langle\Delta\boldsymbol{\theta}\rangle$ as a function of the number of measurement samples $n$. We compare different numbers of parameters $M$ and norms of parameter vector $\abs{\Delta\boldsymbol{\theta}}$. Data is averaged over 10 random instances of parameter $\Delta\boldsymbol{\theta}$ for $N=8$.
	\idg{b} Estimation error as a function of norm of parameter vector $\abs{\Delta\boldsymbol{\theta}}$ for different $M$ and $n$.
	}
	\label{fig:sensing}
\end{figure}

\sectionMain{Potential for metrology}
\revA{We now consider the theoretical potential of the NPQC for general quantum metrology tasks, beyond the restricted protocol we proposed. In particular, we derive the lower bounds of any possible quantum  metrology protocol with the NPQC. 
The accuracy of estimating $\Delta\boldsymbol{\theta}$ as measured by the mean squared error $\text{MSE}(\hat{\boldsymbol{\theta}})=\mathbb{E}(\vert\hat{\boldsymbol{\theta}}-\boldsymbol{\theta}\vert^2)=\text{Tr}[\text{cov}(\hat{\boldsymbol{\theta}})]$ is fundamentally limited by quantum mechanics. For any quantum metrology protocol with $n$ measurements and an unbiased estimator (i.e. $\mathbb{E}(\hat{\boldsymbol{\theta}})=\boldsymbol{\theta}$), the $\text{MSE}(\hat{\boldsymbol{\theta}})\ge\frac{1}{n}\mathcal{F}^{-1}(\boldsymbol{\theta})$ is lower bounded by the inverse of the QFIM, which is called the quantum Cramér-Rao bound~\cite{helstrom1976quantum,liu2019quantum,meyer2021fisher}. 
The NPQC with $\text{Tr}[\mathcal{F}^{-1}(\boldsymbol{\theta}_\text{r})]=M$ has the smallest possible quantum Cramér-Rao bound $Q_\text{min}=M$ for a general class of PQCs constructed from parameterized Pauli rotations and arbitrary unitaries (see Appendix~\ref{sec:QFIM_sup} for derivation)
\begin{equation}\label{eq:MSE}
\text{MSE}(\hat{\boldsymbol{\theta}})\vert_{\boldsymbol{\theta}=\boldsymbol{\theta}_\text{r}}\ge \frac{1}{n}\text{Tr}[\mathcal{F}^{-1}(\boldsymbol{\theta}_\text{r})]=\frac{Q_\text{min}}{n}=\frac{M}{n}\,.
\end{equation}
This can be intuitively understood as for a euclidean QFIM, any variation of the parameters leads to an orthogonal change in the space of quantum states. Thus, each parameter direction is associated with an orthogonal quantum state that can in principle be distinguished from the other states. }

\sectionMain{Discussion}
We introduced the NPQC which features a euclidean quantum geometry with QFIM $\mathcal{F}(\boldsymbol{\theta}_\text{r})=I$ close to the reference parameter $\boldsymbol{\theta}_\text{r}$.
The reference state $\ket{\psi(\boldsymbol{\theta}_\text{r})}$ is completely general and can be any arbitrary quantum state while retaining its euclidean quantum geometry.
The NPQC for $M$ parameters requires only single qubit rotations and $(M-2N)/2$ CPHASE gates, which is a low resource requirement per parameter, comparable with other hardware efficient circuits~\cite{sim2019expressibility}. 

\revB{The NPQC can have barren plateaus with exponentially vanishing gradients, however we find that the decrease with qubit number and depth is much slower compared to other hardware efficient circuits. Note that we compared our results only to the YZ-CNOT circuit, but other commonly used hardware efficient circuits are known to show the same features~\cite{mcclean2018barren,haug2021capacity}. 
Further, the NPQC has an exponentially large effective dimension $D_\text{C}=\text{rank}(\mathcal{F})$~\cite{haug2021capacity} and thus can explore exponentially many directions in Hilbertspace. Yet, its expressibility as measured by the frame potential $F_2$ is lower compared to other hardware efficient circuits, which implies that the NPQC samples states non-uniformly within the Hilbertspace. As high expressibility is linked to the variance of the gradient~\cite{holmes2021connecting}, this may explain why the NPQC has comparatively large gradients.
The differences between NPQC and other hardware efficient circuits may be result of the special structure of the NPQC. In the NPQC, most parameterized rotations act only on odd numbered qubits, and entangling gates only between even and odd numbered qubits. Note that the larger variance of gradients can be advantageous for VQAs, as it allows for training even for relatively deep circuits and higher qubit number compared to other circuits.
We believe other types of NPQCs could be found which have higher expressibility.}

\revA{For VQAs, for the first training step the gradient is equivalent to the QNG, which is known to speed up training~\cite{stokes2020quantum}. While normally the QNG has to be computed, we gain the QNG for free and we can use adaptive learning rates~\cite{haug2021optimal}. We apply our methods to learn quantum states using the fidelity as a cost function.
For the first training iteration, the gradient is exactly equivalent to the QNG. We find that the infidelity is reduced strongly, with better performance for increasing number of qubits. 
As the QFIM is exactly the identity only for the starting point, for further training iterations the gradient differs from the QNG. Still, we find improved performance even for further training steps, as in the vicinity the gradient is still close to the QNG.  This leads to faster training during the first three training steps, and we find better performance compared to alternative training methods.}

We believe our method can also yield speedups for other types of cost functions such as energy. 
The NPQC could also improve the runtime of variational quantum simulation algorithms that require knowledge of the QFIM~\cite{li2017efficient,yuan2019theory,yao2020adaptive}. When studying the short-time dynamics, which is close to the initial state, the QFIM is approximately the identity and we can remove the resource-heavy measurement of the QFIM from these algorithms~\cite{van2020measurement}.

\revA{We provided a protocol to estimate the absolute values of $M$ parameter entries $\Delta\boldsymbol{\theta}$ by sampling in the computational basis. We derived our protocol by employing a first order approximation of the parameter $\Delta \boldsymbol{\theta}$. 
Our protocol becomes more accurate for small $\Delta \boldsymbol{\theta}$, which one could improve further by deriving higher order terms of the expansion.
The sampling can be easily done on NISQ devices and trivially commutes, which allows us to estimate all parameters in parallel. 
Our parameter estimation protocol could be immediately applied in atomic~\cite{bernien2017probing,zhang2017observation} or superconducting setups~\cite{arute2019quantum}. 
One could use our protocol to determine calibration errors in parameterized quantum gates~\cite{cerfontaine2020self}. As advantage, our protocol can measure all parameters at the same time for faster calibration of devices. Future work could study the robustness of multi-parameter estimation against noise in NISQ devices.

For general quantum metrology protocols beyond the restrictions of aforementioned protocol, the accuracy is limited by the quantum Cramér-Rao bound. We showed that for a general class of circuits, NPQCs have the minimal quantum Cramér-Rao bound~\cite{helstrom1976quantum}. This shows the potential of NPQCs for quantum sensing protocols involving many parameters. Future work could find sensing protocols with NISQ-friendly measurement settings that are robust against noise. 
An open question remains whether a protocol that saturates the quantum Cramér-Rao bound exists~\cite{liu2019quantum,meyer2021fisher}.

As further application, the special QFIM of the NPQC allows us to generate arbitrary superposition states of two states. For a desired superposition amplitude, we can compute the corresponding parameters of the NPQC and prepare the state. This scheme could be useful in various state preparation tasks for NISQ computers.}

Finally, we note that a core component of machine learning is information geometry~\cite{abbas2020power,liang2019fisher}. In quantum machine learning based on kernels, the QFIM describes the number of independent features the kernel can represent~\cite{haug2021largescale}. Thus, the NPQC with its special QFIM could be useful for quantum machine learning tasks.

Python code for the numerical calculations are available on Github~\cite{haug2021pqc}.

 \let\oldaddcontentsline\addcontentsline
\renewcommand{\addcontentsline}[3]{}

\medskip
\begin{acknowledgements}
{\noindent {\em Acknowledgements---}} We acknowledge discussions with Kiran Khosla, Christopher Self and Alistair Smith. This work is supported by a Samsung GRC project and the UK Hub in Quantum Computing and Simulation, part of the UK National Quantum Technologies Programme with funding from UKRI EPSRC grant EP/T001062/1. 
\end{acknowledgements}
\bibliography{NaturalCircuit}

\let\addcontentsline\oldaddcontentsline

\clearpage
\appendix

\tableofcontents

\section{Fidelity and variance of NPQC}
We define the fidelity $K_\text{t}(\boldsymbol{\theta})=\abs{\braket{\psi_\text{t}}{\psi(\boldsymbol{\theta})}}^2$ of quantum state $\ket{\psi(\boldsymbol{\theta})}$ in respect to the target state $\ket{\psi(\boldsymbol{\theta}_\text{t})}$.
For small enough parameter distances $\Delta \boldsymbol{\theta}=\boldsymbol{\theta}-\boldsymbol{\theta}'$ the fidelity is approximately described by  a Gaussian~\cite{haug2021optimal}
\begin{equation}\label{eq:kernel_sup}
\mathcal{K}(\boldsymbol{\theta},\boldsymbol{\theta}')=\abs{\braket{\psi(\boldsymbol{\theta})}{\psi(\boldsymbol{\theta}')}}^2\approx\text{exp}[-\frac{1}{4}\Delta \boldsymbol{\theta}^{\text{T}}\mathcal{F}(\boldsymbol{\theta})\Delta \boldsymbol{\theta}]\,.
\end{equation}
Following~\cite{haug2021optimal}, the variance of the gradient can be approximated by
\begin{align*}
&\text{var}(\partial_k K_\text{t}(\boldsymbol{\theta}))=\langle\langle (\partial_k K_\text{t}(\boldsymbol{\theta}))^2\rangle_{\Delta\boldsymbol{\theta}}\rangle_k - \langle\langle \partial_k K_\text{t}(\boldsymbol{\theta})\rangle_{\boldsymbol{\theta}}\rangle_k^2\\
&\approx\frac{1}{M}\frac{\text{Tr}(\mathcal{F}(\boldsymbol{\theta})^2)}{\text{Tr}(\mathcal{F}(\boldsymbol{\theta}))} K_\text{t}(\boldsymbol{\theta})^2\log\left[\frac{ K_0}{K_\text{t}(\boldsymbol{\theta})}\right]\,,\numberthis\label{eq:var_grad_sup}
\end{align*}
where the average is first taken over distance  $\Delta\boldsymbol{\theta}=\boldsymbol{\theta}-\boldsymbol{\theta}_\text{t}$ and then over the gradient indices $k$. We define $K_0=\text{max}_{\boldsymbol{\theta}}\abs{\braket{\psi(\boldsymbol{\theta})}{\psi_t}}^2$ as the maximal possible fidelity for the NPQC. 
For parameter $\boldsymbol{\theta}_\text{r}$, the QFIM is given by $\mathcal{F}(\boldsymbol{\theta}_\text{r})=I$, resulting in simple expressions for fidelity and variance of the gradient respectively.

We now show numerical evidence for the Gaussian form of the fidelity of NPQCs. We show the fidelity as a function of distance between the reference parameter $\boldsymbol{\theta}_\text{r}$ and arbitrarily chosen target parameters $\boldsymbol{\theta}_\text{t}$  $\abs{\Delta\boldsymbol{\theta}_\text{r,t}}^2=\abs{\boldsymbol{\theta}_\text{r}-\boldsymbol{\theta}_\text{t}}^2$ in Fig.\ref{fig:further_data}a. We observe that the data is fitted well with \eqref{eq:kernel_sup} for small distances. For larger distances, it becomes constant and reaches the fidelity $\frac{1}{2^N}$ of Haar random quantum states.
The variance of gradient is shown in Fig.\ref{fig:further_data}b against $\abs{\Delta\boldsymbol{\theta}_\text{r,t}}$. We indeed find a good fit with \eqref{eq:var_grad_sup}. We find that the accuracy of the formulas improve with increasing number of qubits.
\begin{figure}[htbp]
	\centering	
	\subfigimg[width=0.3\textwidth]{a}{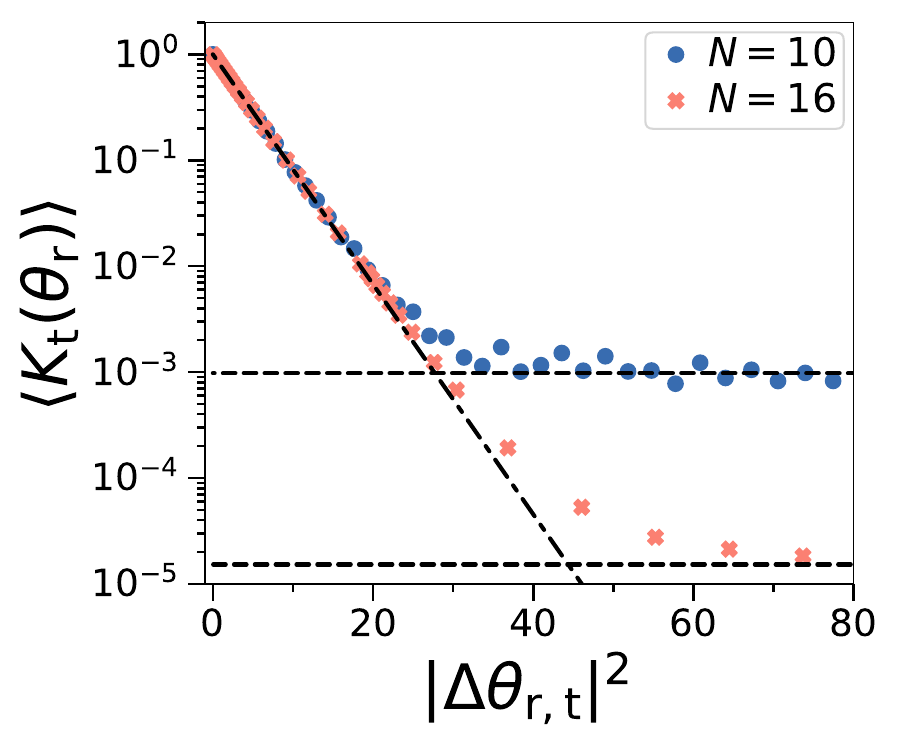}\hfill
	\subfigimg[width=0.3\textwidth]{b}{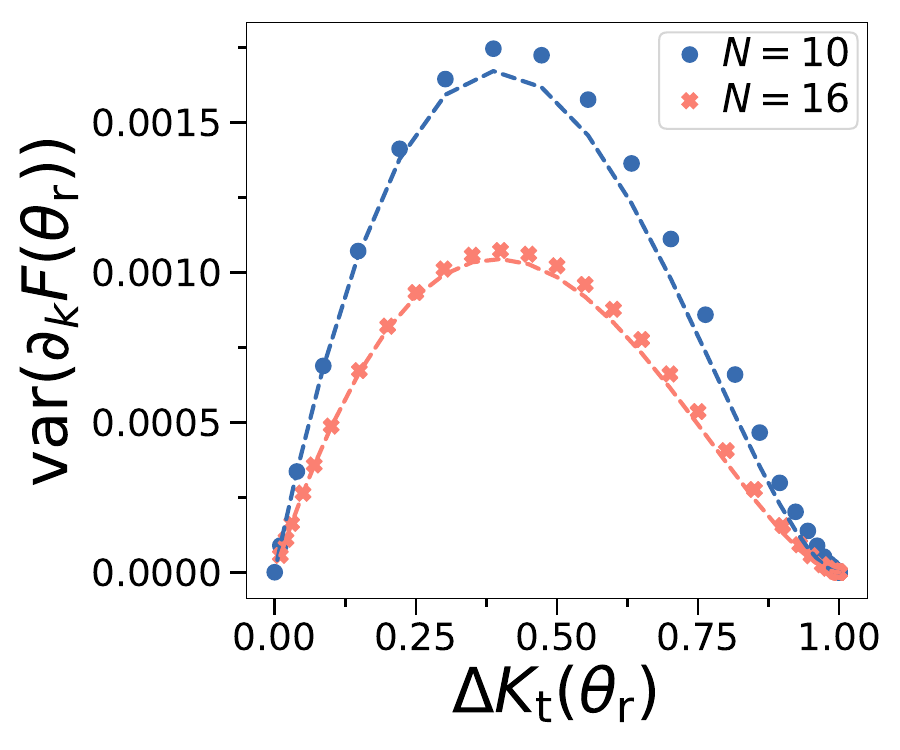}
	\caption{\idg{a} Fidelity $K_\text{t}(\boldsymbol{\theta}_\text{r})$ before optimization as a function of distance between reference and target parameters $\abs{\Delta\boldsymbol{\theta}_\text{r,t}}^2$ for number of layers $p=10$. Dash-dotted line is the theoretic fidelity (\eqref{eq:kernel_sup}). Dashed horizontal lines indicate fidelity of a Haar random state $F(\boldsymbol{\theta}_\text{rand})=\frac{1}{2^N}$.
	\idg{b} Variance of gradient $\text{var}(\partial_kK_\text{t}(\boldsymbol{\theta}_\text{r}))$ as a function of $\abs{\Delta\boldsymbol{\theta}_\text{r,t}}$. Dashed lines are the analytic equations for the variance (\eqref{eq:var_grad_sup}). 
	}
	\label{fig:further_data}
\end{figure}

\section{Gradient update}\label{sec:gradient_sup}
The optimal adaptive learning rates for gradient ascent~\cite{haug2021optimal} can be derived by using that the fidelity is approximately Gaussian for PQCs. These adaptive learning rates tremendously speed up the gradient ascent algorithm.

The goal is to optimise $\boldsymbol{\theta}_\text{t}=\text{argmax}_{\boldsymbol{\theta}}\abs{\braket{\psi(\boldsymbol{\theta})}{\psi_\text{t}}}^2$ for a given target state $\ket{\psi_\text{t}}$.
Initially, we assume $\text{max}_{\boldsymbol{\theta}}\abs{\braket{\psi(\boldsymbol{\theta})}{\psi_\text{t}}}^2=K_0=1$, which means that the PQC is able to represent the state. We relax this condition $K_0<1$ further below.
For an initial parameter $\boldsymbol{\theta}$ we get
a fidelity $K_\text{t}(\boldsymbol{\theta})$.
The gradient ascent algorithm has the update rule for the new parameter $\boldsymbol{\theta}_1$
\begin{equation}\label{eq:grad_update}
\boldsymbol{\theta}_1=\boldsymbol{\theta}+\alpha_1 \nabla K_\text{t}(\boldsymbol{\theta})\,,
\end{equation}
with the learning rate $\alpha_1$.
As we show in \eqref{eq:kernel_sup}, the fidelity follows a Gaussian kernel. We use this to choose $\alpha_1$ such that it is as close as possible to the optimal solution $\boldsymbol{\theta}_1\approx\boldsymbol{\theta}_\text{t}$.
We have
\begin{equation}
K_\text{t}(\boldsymbol{\theta})=e^{-\frac{1}{4}\Delta\boldsymbol{\theta}^{\text{T}}\mathcal{F}(\boldsymbol{\theta})\Delta\boldsymbol{\theta}}\,.
\end{equation}
where we define the distance between target parameter and initial parameter $\Delta \boldsymbol{\theta}=\boldsymbol{\theta}_\text{t}-\boldsymbol{\theta}$. By applying the logarithm we get
\begin{equation}\label{eq:log_kernel}
-4\log(K_\text{t}(\boldsymbol{\theta}))=\Delta\boldsymbol{\theta}^{\text{T}}\mathcal{F}(\boldsymbol{\theta})\Delta\boldsymbol{\theta}\,.
\end{equation}
Reordering \eqref{eq:grad_update} yields
\begin{equation}
\Delta\boldsymbol{\theta}=\alpha_1 \nabla K_\text{t}(\boldsymbol{\theta})\,
\end{equation}
Then, we multiply both sides with $\mathcal{F}^{\frac{1}{2}}(\boldsymbol{\theta})$
\begin{equation}
\mathcal{F}^{\frac{1}{2}}(\boldsymbol{\theta})\Delta\boldsymbol{\theta}=\alpha_1\mathcal{F}^{\frac{1}{2}}(\boldsymbol{\theta}) \nabla K_\text{t}(\boldsymbol{\theta})\,,
\end{equation}
followed by taking square
\begin{equation}
\Delta\boldsymbol{\theta}^{\text{T}}\mathcal{F}(\boldsymbol{\theta})\Delta\boldsymbol{\theta}=\alpha_1^2 \nabla K_\text{t}(\boldsymbol{\theta})^{\text{T}}\mathcal{F}(\boldsymbol{\theta})\nabla K_\text{t}(\boldsymbol{\theta})\,.
\end{equation}
Here, we used $\vert\mathcal{F}^{\frac{1}{2}}\boldsymbol{\mu}\vert^2=\boldsymbol{\mu}^{\text{T}} \mathcal{F}\boldsymbol{\mu}$.
We insert \eqref{eq:log_kernel} and yield
\begin{equation}
\alpha_1=\frac{2\sqrt{-\log(K_\text{t}(\boldsymbol{\theta}))}}{\sqrt{\nabla K_\text{t}(\boldsymbol{\theta})^{\text{T}}\mathcal{F}(\boldsymbol{\theta}) \nabla K_\text{t}(\boldsymbol{\theta})}}\,,
\end{equation}
with the initial update rule
\begin{equation}
\boldsymbol{\theta}_1=\boldsymbol{\theta}+\alpha_1 \nabla K_\text{t}(\boldsymbol{\theta})\,.
\end{equation}
Note we assumed that the PQC is able to represent the target quantum state perfectly $\text{max}_{\boldsymbol{\theta}}\abs{\braket{\psi(\boldsymbol{\theta})}{\psi_\text{t}}}^2=1$. 
We now loosen this restriction.

The target state is defined as
\begin{equation}
\ket{\psi_\text{t}}=\sqrt{K_0}\ket{\psi(\boldsymbol{\theta}_\text{t})}+\sqrt{1-K_0}\ket{\psi_\text{o}}\,,
\end{equation}
where $\ket{\psi_\text{o}}$ is a state orthogonal to any other state that can be represented by the PQC, i.e. $\abs{\braket{\psi_\text{o}}{\psi(\boldsymbol{\theta})}}^2=0\,\,\forall\boldsymbol{\theta}$ and $K_0$ is the maximal fidelity for the target state possible with the PQC.
Then, we find that the initial update rule as defined above is moving in the correct direction, however it overshoots the target parameters. We take this into account via
\begin{equation}
K_\text{t}(\boldsymbol{\theta})=K_0e^{-\frac{1}{4}\Delta\boldsymbol{\theta}^{\text{T}}\mathcal{F}(\boldsymbol{\theta})\Delta\boldsymbol{\theta}}
\end{equation}
\begin{equation}
K_\text{t}(\boldsymbol{\theta}_1)=K_0e^{-\frac{1}{4}(\boldsymbol{\theta}_1-\boldsymbol{\theta}_\text{t})^{\text{T}}\mathcal{F}(\boldsymbol{\theta})(\boldsymbol{\theta}_1-\boldsymbol{\theta}_\text{t})}
\end{equation}
where $K_\text{t}(\boldsymbol{\theta}_1)$ is the fidelity after applying the initial update rule.
The corrected update rule takes the form
\begin{equation}
\boldsymbol{\theta}_\text{t}=\boldsymbol{\theta}+\alpha_\text{t}\nabla K_\text{t}(\boldsymbol{\theta})\,,
\end{equation}
with final learning rate $\alpha_\text{t}$. By subtracting our two update rules we yield
\begin{equation}
\boldsymbol{\theta}_1-\boldsymbol{\theta}_\text{t}=(\alpha_1-\alpha_\text{t})\nabla K_\text{t}(\boldsymbol{\theta})\,.
\end{equation}
We insert above equations into our fidelities
\begin{equation}
K_\text{t}(\boldsymbol{\theta})=K_0 e^{-\frac{1}{4}\alpha_\text{t}^2 \nabla K_\text{t}^{\text{T}}\mathcal{F}\nabla K_\text{t}}
\end{equation}
\begin{equation}
K_\text{t}(\boldsymbol{\theta}_1)=K_0 e^{-\frac{1}{4}(\alpha_1-\alpha_\text{t})^2\nabla K_\text{t}^{\text{T}}\mathcal{F}\nabla K_\text{t}}
\end{equation}
We then divide above equations and solve for $\alpha_\text{t}$
\begin{equation}\label{eq:update_add_sup}
\alpha_\text{t}=\frac{1}{2}\left(\frac{4}{\alpha_1 \nabla K_\text{t}(\boldsymbol{\theta})^{\text{T}}\mathcal{F}(\boldsymbol{\theta})\nabla K_\text{t}(\boldsymbol{\theta})}\log\left(\frac{K_\text{t}(\boldsymbol{\theta}_1)}{K_\text{t}(\boldsymbol{\theta})}\right)+\alpha_1\right)
\end{equation}
with the final update rule
\begin{equation}
\boldsymbol{\theta}_\text{t}'=\boldsymbol{\theta}+\alpha_\text{t}\nabla K_\text{t}(\boldsymbol{\theta})
\end{equation}
where $\boldsymbol{\theta}_\text{t}'$ is the parameter for the PQC after one iteration of gradient ascent with our final update rule.

Note that the adaptive training step requires in general knowledge of the QFIM $\mathcal{F}$. However, for the NPQC at $\boldsymbol{\theta}_\text{r}$, we know that its QFIM $\mathcal{F}(\boldsymbol{\theta}_\text{r})=I$, where $I$ is the identity matrix. By inserting $\mathcal{F}(\boldsymbol{\theta}_\text{r})=I$ into above equations, we get the adaptive learning rates
\begin{align*}
\boldsymbol{\theta}_1=&\boldsymbol{\theta}+\alpha_1 \nabla K_\text{t},\hspace{0.4cm}\alpha_1=\frac{2\sqrt{-\log(K_\text{t}(\boldsymbol{\theta}))}}{\abs{\nabla K_\text{t}(\boldsymbol{\theta})}}\\
\alpha_\text{t}(\boldsymbol{\theta})=&\frac{2}{\alpha_1 \abs{\nabla K_\text{t}(\boldsymbol{\theta})}^2}\log\left(\frac{K_\text{t}(\boldsymbol{\theta}_1)}{K_\text{t}(\boldsymbol{\theta})}\right)+\frac{\alpha_1}{2}\,.\numberthis \label{eq:update_add_euclid_sup}
\end{align*}
Note that this simplification is only valid close to the reference parameter $\boldsymbol{\theta}_\text{r}$. After training for multiple iterations, the parameter $\boldsymbol{\theta}$ will differ from the reference parameter. At this point, the QFIM will become sufficiently non-euclidean such that we cannot assume above update rules anymore. Then, one has to either switch to a heuristic learning rate or calculate the QFIM. We find numerically that it is best to switch after 3 training iterations.

\section{Further training data}\label{sec:training_sub}
Here, we provdie further results on training VQAs with the NPQC.
First, we discuss training as function of number of layers $p$. In Fig.\ref{fig:depth},  we show the infidelity $\langle\Delta K_\text{t}(\boldsymbol{\theta}_\text{t}')\rangle$  after the first training step of adaptive gradient ascent for the NPQC with initial parameter $\boldsymbol{\theta}=\boldsymbol{\theta}_\text{r}$ as a function of number of layers $p$.
We first observe an increase of $\langle\Delta K_\text{t}(\boldsymbol{\theta}_\text{t}')\rangle$ with number of layers $p$, then it reaches a nearly constant level around $p\approx5$ for any $N$ we investigated. Further increase of $p$ yields either a further relatively smaller increase (for $\Delta K_\text{t}(\boldsymbol{\theta}_\text{r})>\xi(N)$) or decrease in $\Delta K_\text{t}(\boldsymbol{\theta}_\text{t}')$ (for $\Delta K_\text{t}(\boldsymbol{\theta}_\text{r})<\xi(N)$), where we numerically find $\xi(10)\approx 0.5$ and $\xi(16)\approx 0.9$.

\begin{figure*}[htbp]
	\centering
	\subfigimg[width=0.3\textwidth]{a}{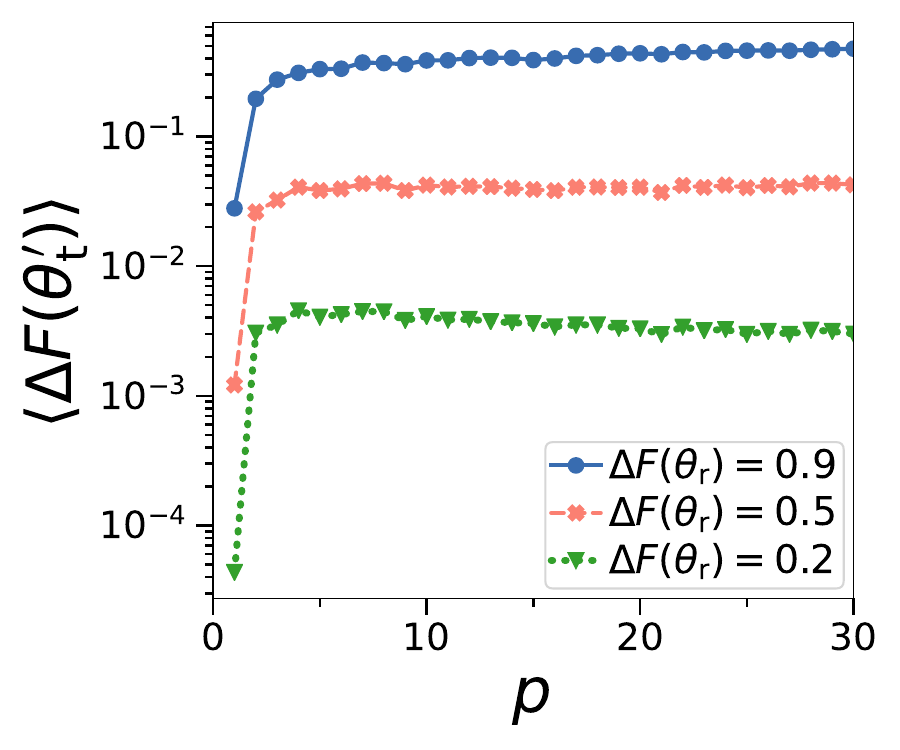}
	\subfigimg[width=0.3\textwidth]{b}{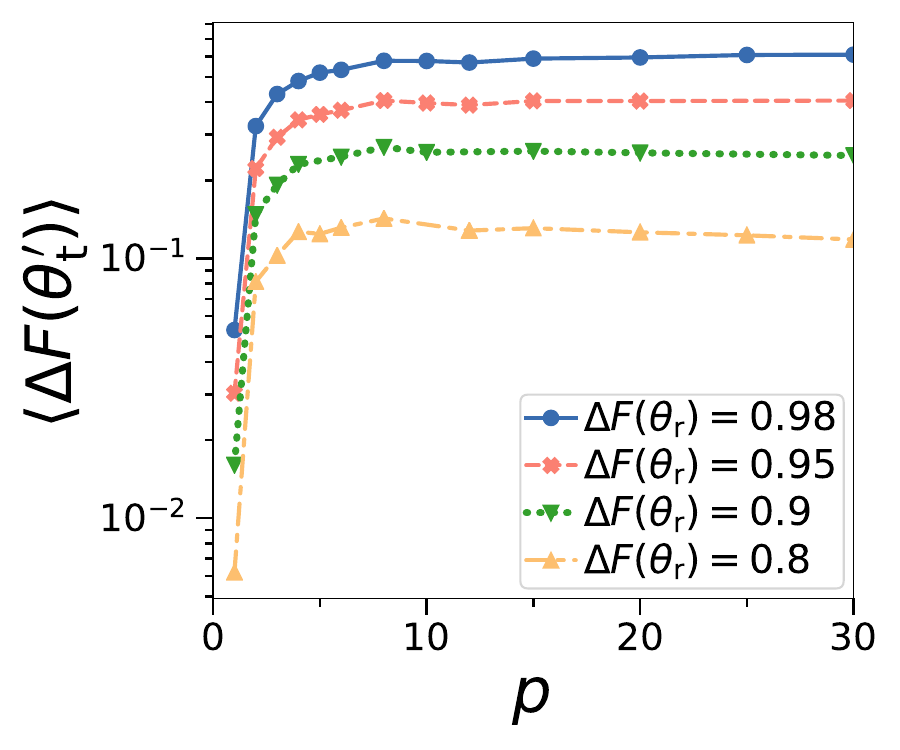}
	\subfigimg[width=0.3\textwidth]{c}{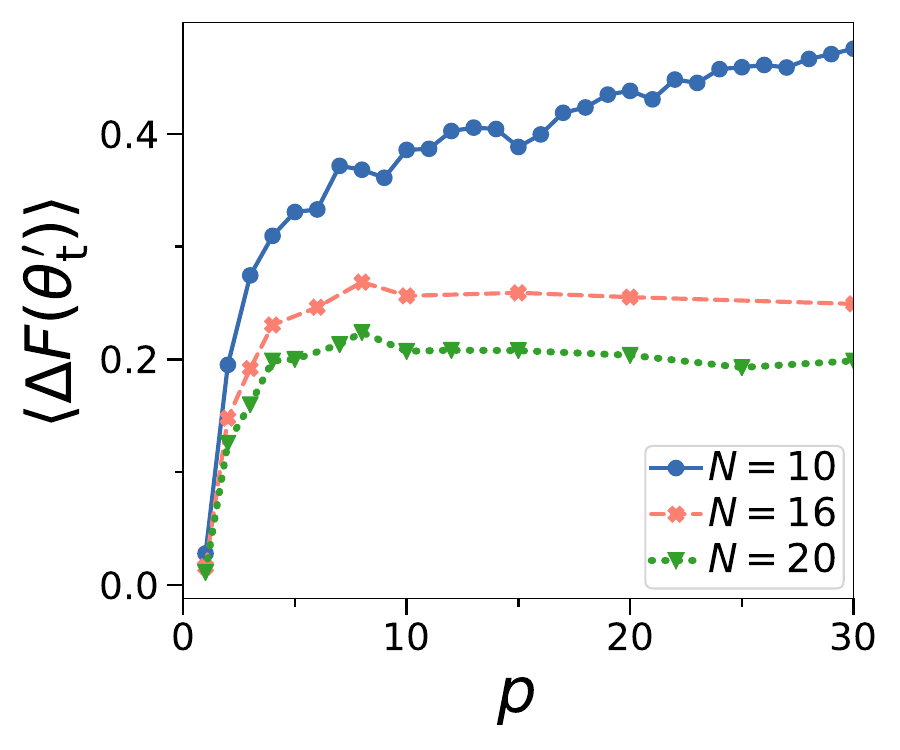}
	\caption{\idg{a} We show average infidelity after optimization $\langle\Delta K_\text{t}(\boldsymbol{\theta}_\text{t}')\rangle$ averaged over 50 random instances against number of layers $p$ for $N=10$ qubits.
	\idg{b} We show infidelity for different initial infidelities $\Delta K_\text{t}(\boldsymbol{\theta}_\text{r})$ for $N=16$ qubits.
	\idg{c} We show infidelity for different number of qubits $N$ for initial infidelity $\Delta K_\text{t}(\boldsymbol{\theta}_\text{r})=0.9$.
	}
	\label{fig:depth}
\end{figure*}

In Fig.\ref{fig:further_training} we show how training depends on the learning rate and initial infidelity. In Fig.\ref{fig:further_training}a, we discuss the infidelity as a function of the learning rate $\lambda$ of gradient ascent. We find that the adaptive learning rate $\alpha_\text{t}$ (\eqref{eq:update_add_sup}) describes the best possible choice of learning rate.
In Fig.\ref{fig:further_training}b, we show the training starting with $\boldsymbol{\theta}_\text{r}$ for a target state with various initial infidelities $\Delta K_\text{t}(\boldsymbol{\theta}_\text{r})$. 
\begin{figure}[htbp]
	\centering
	\subfigimg[width=0.3\textwidth]{a}{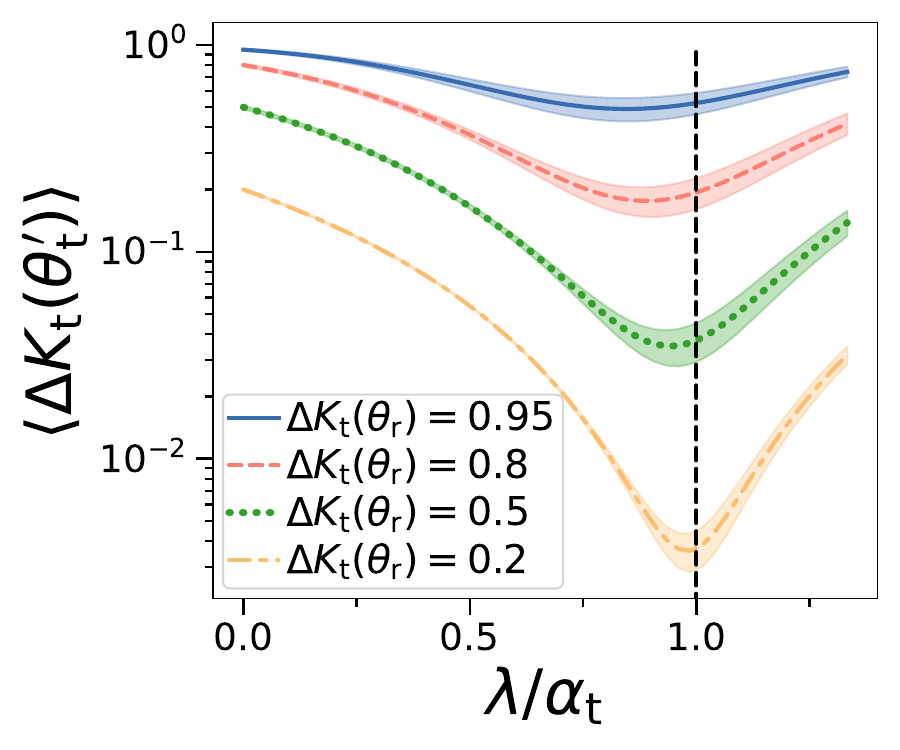}\hfill
	\subfigimg[width=0.3\textwidth]{b}{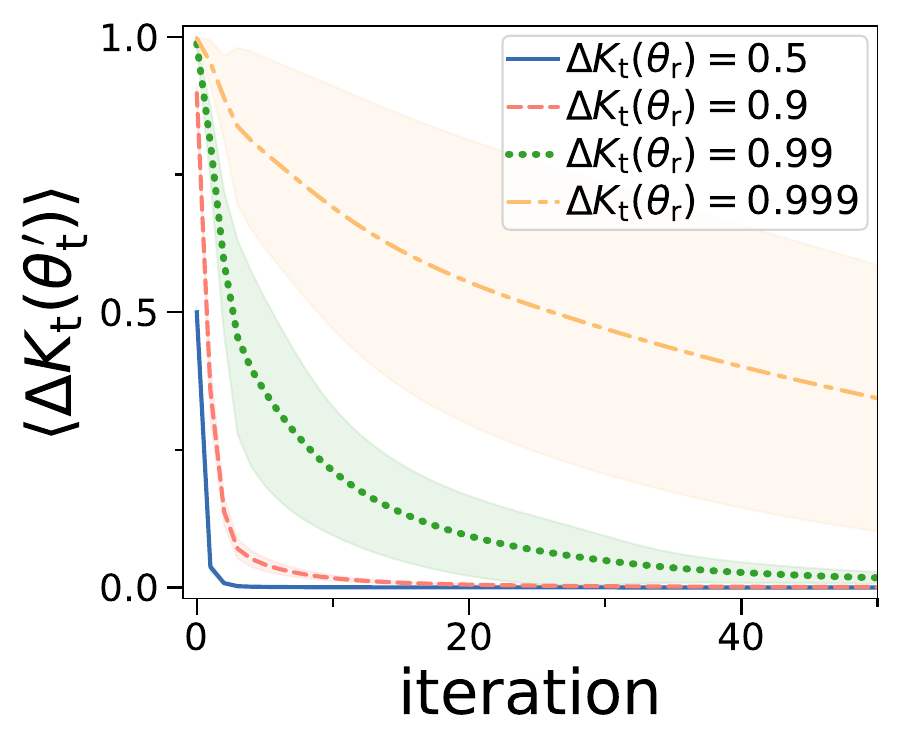}
	\caption{
	\idg{a} Average infidelity after a gradient ascent step $\langle\Delta K_\text{t}(\boldsymbol{\theta}_\text{t}')\rangle$ plotted against learning rate $\lambda$ for single step of gradient ascent. $\lambda$ is normalized in respect to analytically calculated learning rate $\alpha_\text{t}$ (\eqref{eq:update_add_sup}), shown as vertical dashed line. Curves show various infidelity before optimization $\Delta K_\text{t}(\boldsymbol{\theta}_\text{r})$, with shaded area being the standard deviation of $\Delta K_\text{t}(\boldsymbol{\theta}_\text{t}')$. Initial state is $\ket{\psi(\boldsymbol{\theta}_\text{r})}$. The number of qubits is $N=10$ and number of layers $p=10$.
	\idg{b} We show training with NPQC starting at $\boldsymbol{\theta}_\text{r}$ for different initial infidelities $\Delta K_\text{t}(\boldsymbol{\theta}_\text{r})$. We use adaptive learning rate for first three iterations, then use fixed learning rate $\alpha=0.5$.
	}
	\label{fig:further_training}
\end{figure}

\section{NPQC construction for parameter estimation}\label{sec:NPQC_constr_sup}
In Fig.\ref{fig:circuit_sensing}, we show the modified NPQC $U_y(\boldsymbol{\theta})\ket{0}$ used for multi-parameter sensing. Compared to the original NPQC, we remove the $z$ rotations and fix the $y$ rotations on the qubits with even index to $\pi/2$. 
The final quantum state used for sensing is then given by $\ket{\psi_y(\boldsymbol{\theta})}=U_y^\dagger(\boldsymbol{\theta}_\text{r})U_y(\boldsymbol{\theta})\ket{0}$, where we added the adjoint of the modified NPQC with fixed parameter $\boldsymbol{\theta}_\text{r}$.
\begin{figure}[htbp]
	\centering
	\subfigimg[width=0.4\textwidth]{}{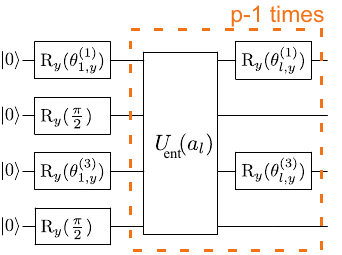}
	\caption{Modified NPQC $U_y(\boldsymbol{\theta})\ket{0}$ for multi-parameter sensing. $U_\text{ent}(a_l)$ is defined in the main text.
	}
	\label{fig:circuit_sensing}
\end{figure}

\section{NPQC approximation for parameter estimation}\label{sec:sensing_approx}
\revB{We now motivate that the NPQC modified for parameter estimation can be approximated with \eqref{eq:sense_state}.

First, we vary $\ket{\psi_y(\boldsymbol{\theta}_\text{r}+\beta\boldsymbol{e}_i)}$ by a factor $\beta$ along the unit vector $\boldsymbol{e}_i$ for a single parameter entry $i$, while keeping all other parameters constant. 
We rewrite the unitary $U_y(\boldsymbol{\theta}_\text{r}+\beta\boldsymbol{e}_i)=U_\text{B}\exp(-i\beta\sigma^y_{g(i)})U_\text{A}$, where $U_\text{A}$, $U_\text{B}$ includes all gates before and after the rotation for parameter $i$, and $g(i)$ is the index of the qubit for the rotation corresponding to the $i$th parameter. The derivative is given by
$\partial_i U_y(\boldsymbol{\theta}_\text{r}+\beta\boldsymbol{e}_i)=U_\text{B}(-i\sigma^y_{g(i)})\exp(-i\beta\sigma^y_{g(i)})U_\text{A}$.

For infinitesimal small variations $\beta$, we can write 
\begin{align*}
&\frac{1}{\beta}(\ket{\psi_y(\boldsymbol{\theta}_\text{r}+\beta\boldsymbol{e}_i)}-\ket{\psi_y(\boldsymbol{\theta}_\text{r})})\big\vert_{\beta\rightarrow 0}=\\
&\partial_i\ket{\psi_y(\boldsymbol{\theta})}\big\vert_{\boldsymbol{\theta}=\boldsymbol{\theta}_\text{r}}=U_y^\dagger(\boldsymbol{\theta}_\text{r})\partial_i U_y(\boldsymbol{\theta})\ket{0}=\\
&U_\text{A}^\dagger\exp(i\beta\sigma^y_{g(i)})U_\text{B}^\dagger U_\text{B}(-i\sigma^y_{g(i)})\exp(-i\beta\sigma^y_{g(i)})U_\text{A}\ket{0}=\\
&U_\text{A}^\dagger(-i\sigma^y_{g(i)})U_\text{A}\ket{0}\,,
\end{align*}
where $U_\text{A}$, $U_\text{B}$ are Clifford unitaries as they are only composed of Clifford gates. Any Clifford unitary $U_\text{A}$ maps a Pauli string $P$ to another Pauli string $P'=U_\text{A}^\dagger P U_\text{A}$. Further, $U_\text{A}$ is a real unitary and thus the resulting state must be a real quantum state. A Pauli string applied to the computational basis state $\ket{0}$ yields another computational basis state $P'\ket{0}=(-1)^{\alpha_i}\ket{v_i}$ with basis $v_i$ and some phase factor $(-1)^{\alpha_i}$, $\alpha_i\in\{0,1\}$. Note that in our case $P'$ and thus the phase factors are real valued.  As such, we find \begin{equation}
\partial_i\ket{\psi_y(\boldsymbol{\theta})}=\frac{1}{\beta}(\ket{\psi_y(\boldsymbol{\theta}_\text{r}+\beta\boldsymbol{e}_i)}-\ket{\psi_y(\boldsymbol{\theta}_\text{r})})\big\vert_{\beta\rightarrow 0}=(-1)^{\alpha_i}\ket{v_i}\,.
\end{equation}
For small $\beta$, we can write the state without normalization
\begin{equation}\label{eq:variation}
\ket{\psi_y(\boldsymbol{\theta}_\text{r}+\beta\boldsymbol{e}_i)}\sim\ket{\psi_y(\boldsymbol{\theta}_\text{r})}+\beta\partial_i\ket{\psi_y(\boldsymbol{\theta})}\sim\ket{0}+(-1)^{\alpha_i}\beta\ket{v_i}\,.
\end{equation}

Now, we are varying not only one, but the $M$-dimensional parameter $\Delta\boldsymbol{\theta}$ with $\ket{\psi_y(\boldsymbol{\theta}_\text{r}+\Delta\boldsymbol{\theta})}$. As $\mathcal{F}(\boldsymbol{\theta}_\text{r})=I$, this implies that for small $\Delta\boldsymbol{\theta}$ all variations can be treated independent of each other. Further, we know that the fidelity for small variations follows \eqref{eq:fidelityQFIM}. Combining these two conditions and \eqref{eq:variation}, we find that the state for $\vert\Delta\boldsymbol{\theta}\vert\ll 1$ is given by
\begin{equation}\label{eq:sense_state_sup}
\ket{\psi_y(\boldsymbol{\theta}_\text{r}+\Delta\boldsymbol{\theta})}\approx\sqrt{1-\frac{1}{4}\vert\Delta\boldsymbol{\theta}\vert^2}\ket{0}+\frac{1}{2}\sum_{i=1}^M (-1)^{\alpha_i}\Delta\boldsymbol{\theta}_i\ket{v_i}\,.
\end{equation}
}

\section{Cramér-Rao bound of NPQC}\label{sec:QFIM_sup}
We assume a general class of PQCs composed of arbitrary unitaries and Pauli rotations. We have $\ket{\psi(\boldsymbol{\theta})}=U(\boldsymbol{\theta})\ket{0}=\prod_{l=1}^M U_l(\theta_l)\ket{0}$ given by $M$ layers and $M$-dimensional parameter vector $\boldsymbol{\theta}$. The unitary at layer $l$ is given by $U_l(\theta_l)=R_l(\theta_l)W_l$, with constant $N$-qubit unitary $W_l$ and a parameterized unitary $R_l(\theta_l)=\exp(-i\frac{\theta_l}{2}P_l)$, with parameter $\theta_l$ and Pauli string $P_l=\otimes_{k=1}^N \boldsymbol{\sigma}$, where $\boldsymbol{\sigma}\in\{\sigma^x,\sigma^y,\sigma^z,I\}$ is either a Pauli matrix or the identity. The NPQC belongs to this class of PQC as well as commonly used hardware efficient PQCs.

The quantum Fisher information metric $\mathcal{F}$ is a $M$ dimensional positive semidefinite matrix given by $\mathcal{F}_ {ij}(\boldsymbol{\theta})=4[\braket{\partial_i \psi}{\partial_j \psi}-\braket{\partial_i \psi}{\psi}\braket{\psi}{\partial_j \psi}]$, where $\partial_j\ket{\psi}$ is the gradient in respect to parameter $j$.

The quantum Fisher information metric (via the quantum Cramér-Rao bound) sets a lower bound on the error made when using a quantum system as a sensor~\cite{meyer2021fisher}.
\begin{theorem}\label{theorem:NPQC}
For above defined class of PQCs, the minimal possible mean squared error for the unbiased estimator $\hat{\boldsymbol{\theta}}$ is given by
\begin{equation}\label{eq:MSE_sup}
\text{MSE}(\hat{\boldsymbol{\theta}})\ge\frac{1}{n}\text{Tr}[\mathcal{F}^{-1}(\boldsymbol{\theta})]\ge\frac{Q_\text{min}}{n}=\frac{M}{n}\,,
\end{equation}
where $n$ is the number of measurements performed and $Q_\text{min}=M$.
\end{theorem}
The NPQC assumes this lower bound as $\text{Tr}[\mathcal{F}^{-1}(\boldsymbol{\theta}_\text{r})]=M$, making the NPQC one of the optimal sensors within above defined class of PQCs.
\begin{proof}
We now proceed to proof Theorem \ref{theorem:NPQC}.
The derivative acting on the unitary in layer $l$ is given by $\partial_i U_l=\delta_{il}(-i\frac{1}{2}P_l)U_l$, where $\delta_{il}$ is the Kronecker delta. 
We define $U_{[l_1:l_2]}=U_{l_2}\dots U_{l_1}$, with $l_2\ge l_1$ and $\ket{\psi_l}=U_{[1:l]}\ket{0}$. 
With this notation, we find
$\ket{\partial_l \psi}=U_{[l+1:M]}(\partial_l U_l) U_{[1:l-1]}\ket{0}=U_{[l+1:M]}(-i\frac{1}{2}P_l)U_l U_{[1:l-1]}\ket{0}=U_{[l+1:M]}(-i\frac{1}{2}P_l) U_{[1:l]}\ket{0}$.
We can now compute
$\braket{\partial_l \psi}{\partial_l \psi}=\frac{1}{4}$ due to $P_l^2=I$ and $\braket{\psi}{\partial_l \psi}=-i\frac{1}{2}\bra{\psi_l}P_l\ket{\psi_l}$.
The diagonal terms of the quantum Fisher information metric are then given by
\begin{equation}
\mathcal{F}_{ii}=4[\braket{\partial_i \psi}{\partial_i \psi}-\braket{\partial_i \psi}{\psi}\braket{\psi}{\partial_i \psi}]=1-(\bra{\psi_i}P_i\ket{\psi_i})^2\,.
\end{equation}
Due to the eigenvalues of $P_i$ being $\pm 1$, we have $0\le(\bra{\psi_i}P_i\ket{\psi_i})^2\le1$.
Thus, the diagonal entries $\mathcal{F}_{ii}$ are within $0\le\mathcal{F}_ {ii}\le 1$ and the trace of $\mathcal{F}$ is upper bounded by
\begin{equation}\label{eq:f_trace_bound}
\text{Tr}(\mathcal{F})=\sum_{i=1}^M \mathcal{F}_{ii}\le M\,.
\end{equation}
By combining \eqref{eq:f_trace_bound} and Lemma \ref{lemma:inv_bound}, \eqref{eq:MSE_sup} follows immediately.
\end{proof}
\begin{lemma}\label{lemma:inv_bound}
Given a positive semdefinite matrix $\mathcal{A}$ with dimension $M\times M$, $M\in\mathbb{N}$, the trace of the inverse matrix $\mathcal{A}^{-1}$ is lower bounded by $\text{Tr}(\mathcal{A}^{-1})\ge \frac{M^2}{\text{Tr}(\mathcal{A})}$.
\end{lemma}
\begin{proof}
For a sequence of numbers $x_1,x_2,\dots,x_M$ with $x_n\ge0$, the arithmetic mean is always larger than harmonic mean
\begin{equation}
\frac{1}{M}\sum_{n=1}^M x_n \ge \frac{M}{\sum_{n=1}^M \frac{1}{x_n}} \,,
\end{equation}
which is known from the relations of the Pythagorean means. A simple calculation shows
\begin{equation}\label{eq:relation_mean}
\sum_{n=1}^M \frac{1}{x_n}\ge \frac{M^2}{\sum_{n=1}^M x_n}\,.
\end{equation}
The positive semidefinite matrix $A$ has only non-negative eigenvalues $\lambda_n\ge0$. The trace is given by 
\begin{equation}
\text{Tr}(\mathcal{A})=\sum_{n=1}^M\lambda_n\,.
\end{equation}
and accordingly for the inverse
\begin{equation}
\text{Tr}(\mathcal{A}^{-1})=\sum_{n=1}^M\frac{1}{\lambda_n}\,.
\end{equation}
Using \eqref{eq:relation_mean}, we can immediately show
\begin{equation}\label{eq:trace_bound}
\text{Tr}(\mathcal{A}^{-1})\ge \frac{M^2}{\text{Tr}(\mathcal{A})}\,,
\end{equation}
\end{proof}
Using \eqref{eq:trace_bound} and \eqref{eq:f_trace_bound}, we find
\begin{equation}
\text{Tr}(\mathcal{F}^{-1})\ge \frac{M^2}{\text{Tr}(\mathcal{F})}\ge M\,,
\end{equation}
which we insert in \eqref{eq:MSE_sup}.

\end{document}